\documentclass[a4paper,UKenglish,cleveref, autoref, thm-restate]{lipics-v2019}

\title{Non-uniform complexity via non-wellfounded proofs} 

\titlerunning{Non-uniformity via non-wellfoundedness} 

\author{Gianluca Curzi}{University of Birmingham, UK \and  \url{http://gianlucacurzi.com} }{g.curzi@bham.ac.uk}{}{}

\author{Anupam Das}{University of Birmingham, UK \and  \url{https://anupamdas.com}. }{a.das@bham.ac.uk}{}{}

\authorrunning{G. Curzi and A. Das} 

\Copyright{Gianluca Curzi and Anupam Das} 



\category{} 

\relatedversion{} 

\supplement{}

\funding{Both authors are supported by a UKRI Future Leaders Fellowship, \emph{Structure vs.\ Invariants in Proofs}, project reference MR/S035540/1.}

\acknowledgements{\Achange{We thank the anonymous reviewers for their helpful comments and suggestions.}}

\nolinenumbers 

\EventEditors{Bartek Klin and Elaine Pimentel}
\EventNoEds{2}
\EventLongTitle{31st EACSL Annual Conference on Computer Science Logic (CSL 2023)}
\EventShortTitle{CSL 2023}
\EventAcronym{CSL}
\EventYear{2023}
\EventDate{February 13--16, 2023}
\EventLocation{Warsaw, Poland}
\EventLogo{}
\SeriesVolume{252}
\ArticleNo{25}

\usepackage{hyperref}

\usepackage[utf8]{inputenc}

\usepackage{tikz}
\usetikzlibrary{positioning,calc,arrows,shapes,  tikzmark, decorations.pathreplacing}

\usepackage{virginialake}
\usepackage{tikz-cd}
\usepackage{cleveref}
 \usepackage[textsize=tiny
  , disable
 ]{todonotes}

\usepackage{caption}
\usepackage{subcaption}

 \usepackage{enumerate}

\usepackage{cmll}
\usepackage{amsmath}
\usepackage{amsthm}
\usepackage{amssymb}
\usepackage{cleveref}
\usepackage{todonotes}

\newcommand{\dfn}{:=}
\renewcommand{\emptyset}{\varnothing}
\renewcommand{\setminus}{-}
\renewcommand{\epsilon}{\varepsilon}

\newcommand{\red}[1]{{\color{red}#1}}
\newcommand{\blue}[1]{{\color{blue}#1}}
\newcommand{\purple}[1]{{\color{purple}#1}}
\definecolor{mygreen}{rgb}{0, 0.5, 0}
\newcommand{\green}[1]{{\color{mygreen}#1}}
\newcommand{\orange}[1]{{\color{orange}#1}}

\newcommand{\black}[1]{{\color{black}#1}}

\newcommand{\anupam}[1]{\todo[color=red]{A: #1}}

\newcommand{\Achange}[1]{\black{#1}}

\newtheorem{thm}{Theorem}
\newtheorem{lem}[thm]{Lemma}
\newtheorem{prop}[thm]{Proposition}
\newtheorem{cor}[thm]{Corollary}

\theoremstyle{definition}
\newtheorem{defn}[thm]{Definition}
\newtheorem{exmp}[thm]{Example}
\newtheorem{conv}[thm]{Convention}

\theoremstyle{remark}
\newtheorem{rem}[thm]{Remark}

\newcommand{\Nat}{\mathbb{N}}

\newcommand{\RR}{\mathbb R}

\newcommand{\nor}[1]{#1_{1;0}}
\newcommand{\safnor}[1]{#1_{1;1}}

\newcommand{\RRnor}{\nor\RR}
\newcommand{\RRsafnor}{\safnor\RR}

\newcommand{\permpref}{\subset}
\newcommand{\permprefeq}{\subseteq}

\newcommand{\n}{N}

\newcommand{\sq}{\Box}
\newcommand{\sqn}{{\sq\n}}

\newcommand{\sn}{\sqn}
\newcommand{\lists}[3]{\overset{#1}{\overbrace{{#2},{\ldots},{#3}}} }

\newcommand{\ptime}{\mathbf{P}}

\newcommand{\elementary}{\mathbf{ELEMENTARY}}

\newcommand{\f}{\mathbf{F}}
\newcommand{\fptime}{\f\ptime}

\newcommand{\felementary}{\f\elementary}

\newcommand{\FPpoly}{\mathbf{FP} / \mathit{poly}}
\newcommand{\Ppoly}{\mathbf{P} / \mathit{poly}}
\newcommand{\Lpoly}{\mathbf{L} / \mathit{poly}}
\newcommand{\FEpoly}{\felementary / \mathit{poly}}

\renewcommand{\succ}[1]{\mathsf{s}_{#1}}
\newcommand{\suc}[1]{\mathsf{s}_{#1}}
\newcommand{\pred}{\mathsf{p}}

\newcommand{\cnd}{\mathsf{cond}}

\newcommand{\s}[1]{\vert {#1}\vert}

\newcommand{\model}[1]{\denot{#1}}
\newcommand{\denot}[1]{f_{#1}}

\newcommand{\sumlen}[1]{|\hspace{-.15em}|#1|\hspace{-.15em}|}

\newcommand{\sexp}{\mathcal{E}}
\newcommand{\cconc}{\mathcal{C}}

\newcommand{\incrementation}{\mathcal{I}}

\newcommand{\completeness}{\mathcal{F}}
\newcommand{\primrec}{\mathcal{R}}

\newcommand{\gfunction}{\mathcal{G}}
\newcommand{\hfunction}{\mathcal{H}}

\newcommand{\numeral}[1]{\underline{#1}}
\newcommand{\permi}[1]{\vec{#1}}

\newcommand{\saferec}{\mathsf{srec}}
\newcommand{\srec}{\saferec}
\newcommand{\saferecwfpo}[1]{\saferec_{#1}}
\newcommand{\srecwfpo}[1]{\saferecwfpo{#1}}

\newcommand{\srecpp}{\srecwfpo{\permpref}}

\newcommand{\ssrecwfpo}[1]{\mathsf{ssrec}_{#1}}

\newcommand{\ssrecpp}{\ssrecwfpo{\permpref}}

\newcommand{\bc}{\mathsf{B}}
\newcommand{\B}{\bc}
\newcommand{\bcwfpo}[1]{\bc^{#1}}
\newcommand{\bcnorec}{\bc^{-}}

\newcommand{\bcpp}{\bcwfpo{\permpref}}

\newcommand{\seqar}{\Rightarrow}
\newcommand{\der}{\mathcal{D}}
\newcommand{\pder}[1]{\der_{#1}}
\newcommand{\rd}{{R}_\der}
\newcommand{\derrd}{\langle \der, \rd \rangle}
\newcommand{\bud}{\mathit{Bud}}

\newcommand{\close}[1]{{C}_{#1}}
\newcommand{\open}[1]{{O}_{#1}}
 
 \newcommand{\hnu}{\nu_0}

\newcommand{\id}{\mathsf{id}}
\newcommand{\wk}{\mathsf{w}}
\newcommand{\exch}{\mathsf{e}}

\newcommand{\cut}{\mathsf{cut}}

\newcommand{\boxlef}{\sq_l}
\newcommand{\boxrig}{\sq_r}
\newcommand{\sql}{\boxlef}
\newcommand{\sqr}{\boxrig}
\newcommand{\zero}{0}
\newcommand{\unit}{1}
\newcommand{\com}{\mathsf{dis}}
\newcommand{\rules}{\mathsf{r}}
\newcommand{\relaxed}{\{\cnd_\sq,\cnd_\n, \suc 0, \suc 1, \id\}}

\newcommand{\cyclic}{\mathsf{C}}
\newcommand{\cbc}{\cyclic\bc}
\newcommand{\ncbc}{\cyclic\mathsf{NB}}

\newcommand{\nucbc}{\mathsf{nu}\bc}

\newcommand{\cndl}{|\cnd|}

\begin{document}

\maketitle

\begin{abstract}
Cyclic and non-wellfounded proofs are now increasingly employed to establish metalogical results in a variety of settings, in particular for type systems with forms of (co)induction. Under the Curry-Howard correspondence, a cyclic proof can be seen as a typing derivation `with loops', closer to low-level machine models, and so comprise a highly expressive computational model that nonetheless enjoys excellent metalogical properties.

In recent work, we showed how the cyclic proof setting can be further employed to model computational complexity, yielding characterisations of the polynomial time and elementary computable functions. These characterisations are `implicit', inspired by Bellantoni and Cook's famous algebra of safe recursion, but exhibit greater expressivity thanks to the looping capacity of cyclic proofs.

In this work we investigate the capacity for \emph{non-wellfounded   } proofs, where finite presentability is relaxed, to model non-uniformity in complexity theory. In particular, we present a characterisation of the class $\FPpoly$ of functions computed by polynomial-size circuits. While relating non-wellfoundedness to non-uniformity is a natural idea, the precise amount of irregularity, informally speaking, required to capture $\FPpoly$ is given by proof-level conditions novel to cyclic proof theory. Along the way, we formalise some (presumably) folklore techniques for characterising non-uniform classes in relativised function algebras with appropriate oracles.
\end{abstract}

	\section{Introduction}

\emph{Non-wellfounded proof theory} is the study of possibly infinite (but finitely branching) proofs, where    appropriate global correctness criteria guarantee logical consistency. 
This area originates (in its modern guise) in the context of the modal $\mu$-calculus~\cite{niwinski1996games,dax2006proof}, serving as an alternative framework to manipulate least and greatest fixed points, and hence to model inductive and coinductive reasoning. 
Since then, non-wellfounded proofs have been widely investigated in many respects, such as predicate logic~\cite{brotherston2011sequent, BerardiT19},  algebras~\cite{das2017cut,DP18},   arithmetic~\cite{Simpson17, BerardiT17, das2018logical},  proofs-as-programs  interpretations~\cite{baelde2016infinitary,DeS19,Das2021-preprint,Kuperberg-Pous21, Das2021}, and continuous cut-elimination~\cite{mints1978finite, fortier2013cuts}.  Special attention in these works is drawn to \emph{cyclic} (or \emph{regular}) proofs, i.e.~non-wellfounded proofs with only finitely many distinct subproofs, 
comprising a natural notion of finite presentability in terms of (possibly cyclic) directed graphs. 

The  \emph{Curry-Howard} reading of non-wellfounded proofs has revealed a deep connection between proof-theoretic properties and computational behaviours~\cite{Das2021-preprint,Kuperberg-Pous21, Das2021}. On the one hand, the typical correctness conditions ensuring consistency, called \emph{\Achange{progressing}} (or \emph{validity}) criteria,  correspond to totality:  functions computed by progressing proofs are always well-defined on all arguments. 
On the other hand, regularity has a natural counterpart in the notion of  \emph{uniformity}: circular  proofs can be properly regarded as programs,  i.e.~as finite sets of machine instructions, thus having a `computable' behaviour.

In a recent work~\cite{CurziDas}, the authors extended these connections between non-wellfounded proof theory and computation to the realm of \emph{computational complexity}.
We introduced the proof systems $\cbc$ and $\ncbc$ capturing, respectively,  the class of functions computable in polynomial time ($\fptime$) and the elementary functions ($\felementary$). These proof systems are defined by identifying global conditions on circular progressing proofs motivated by ideas from \emph{Implicit Computational Complexity} (ICC). 

ICC, broadly construed, is the study of machine-free (and often bound-free) characterisations of complexity classes. 
One of the seminal works in the area is Bellantoni and Cook's function algebra $\bc$ for $\fptime$ based on \emph{safe recursion} \cite{BellantoniCook}.
 The prevailing idea behind safe recursion (and its predecessor, \emph{ramified recursion} \cite{Leivant91}) is to partition function arguments  into `safe' and `normal', namely writing $f(x_1, \dots, x_{m}; y_1, \dots, y_{n})$ when $f$ takes $m$ normal inputs $\vec x$ and $n$ safe inputs $\vec y$. In functions of $\bc$, the recursive parameters are always normal arguments, while recursive calls can only appear in safe position; 
 hence, no recursive call can be used as recursive parameters of other previously defined functions.
Our system $\cbc$ morally represents a cyclic proof theoretic formulation of $\bc$. 

To establish the characterisation result for $\cbc$ we developed a novel function algebra for $\fptime$, called $\bcpp$. Roughly, the latter  extends $\bc$ with a more expressive recursion mechanism on a special well-founded preorder,  `$\permpref$',  based on permutation of prefixes of normal arguments, and whose definition requires relativisation of the algebra to admit \emph{oracles}. 
The characterisation theorem is then obtained by a `sandwich' technique, where the function algebras $\bc$ and $\bcpp$ serve, respectively, as lower and upper bounds for $\cbc$.

In this paper we investigate the  computational interpretation of more general  \emph{non-wellfounded} proofs, where finite presentability is relaxed in order to model \emph{non-uniform complexity}. 
In particular we consider the class $\FPpoly$ of functions computable in polynomial time by Turing machines with access to \emph{polynomial advice}.
Equivalently, $\FPpoly$ is the class of functions computed by families of polynomial-size circuits.
Note, in particular, that \Achange{$\FPpoly$} includes \emph{undecidable problems}, and so cannot be characterised by purely cyclic proof systems or usual function algebras, which typically have only computable functions.

We define the system $\nucbc$ (`non-uniform $\bc$'), allowing a form of non-wellfoundedness somewhere between arbitrary non-wellfounded proofs and full regularity, and show that $\nucbc$ duly characterises $\FPpoly$.
The characterisation theorem for $\nucbc$ relies on an  adaption of the aforementioned sandwich technique for  $\cbc$ to the current setting. 
This requires a  relativisation of both $\bc$ and $\bcpp$ to a set of oracles, which we call $\RR$, deciding properties of string length. 
As a byproduct of our proof method we also obtain new relativised function algebras for $\FPpoly$ based on safe recursion, $\bc(\RR)$ and $\bcpp(\RR)$; these are folklore-style results that, as far as we know, have not yet appeared in the literature. 

The overall structure of our result relies on a `grand tour' of inclusions, summarised as:
\[
\FPpoly 
 \overset{\text{\tiny P.\ref{prop:fppoly-in-B(R0)}}}{\subseteq}
\bc (\RRnor) 
 \overset{\text{\tiny P.\ref{prop:b(f)-in-cbc(f)}}}{\subseteq}
\cbc (\RRnor) 
 \overset{\text{\tiny P.\ref{prop:cbc(R0)-to-nuB}}}{\subseteq}
\nucbc 
 \overset{\text{\tiny T.\ref{lem:non-wellfounded-oracles}}}{\subseteq}
\cbc(\RRsafnor) 
 \overset{\text{\tiny L.\ref{lem:rel-tran-lem}}}{\subseteq}
\bcpp(\RRsafnor)
 \overset{\text{\tiny P.\ref{prop:relativised-characterisation}}}{\subseteq}
\fptime(\RR)
 \overset{\text{\tiny P.\ref{prop:fppoly=fptime(RR)}}}{\subseteq}
\FPpoly
\]
While this may seem like a long route to take, the structure of our argument is designed so that each of the above inclusions are relatively simple to establish and, as we said, yields several intermediate characterisions of $\FPpoly$ of self-contained interest.

\noindent
\textbf{Related work.} Characterisations of non-uniform complexity classes in the style of ICC have been considered in the context of the $\lambda$-calculus~\cite{Mazza14}  and  variants of linear logic~\cite{MazzaT15-non-uniform}. The former captures the class $\Ppoly$, i.e., the languages decided  by families of polynomial circuits, while the latter also captures $\Lpoly$, i.e., the languages decided by families of polynomial size branching
programs (i.e.\ decision trees with sharing). 
However this is the first work (as far as we know) that attempts to relate non-wellfoundedness in proof theory to non-uniformity in complexity theory.

The relativised proof systems and function algebras presented in this paper only query `bits of real numbers'. 
Proof systems based on linear logic and  implementing polytime computation over actual binary streams have been considered, e.g., in~\cite{HainryMP20}, which provide an ICC-like  characterisation of Ko's class of polynomial time computable functions over real numbers~\cite{Ko}.

\noindent
\textbf{Outline of the paper.} This paper is structured as follows. In~\Cref{sec2} we recall some preliminaries on non-uniform and implicit complexity, in particular a proof theoretic formulation of the algebra $\bc$.
In \cref{sec:cbc-nucbc} we recall the circular system $\cbc$ from \cite{CurziDas}, and introduce our new system $\nucbc$. 
In~\Cref{sec3} we take an interlude to present some \emph{relativised} characterisations of $\FPpoly$, both in the machine setting and the implicit setting, that will later serve use in our grand tour of inclusions. 
In~\Cref{sec4} we employ those characterisations to establish the lower bound for $\nucbc$, and in
\cref{sec:nuB-in-cb(R)} we recast $\nucbc$ as a sort of relativised circular system.
Finally in \cref{sec:cb(R)-in-fppoly} we adapt results from \cite{CurziDas} translating circular proofs to an appropriate function algebra to the relativised setting, thereby achieving the upper bound for $\nucbc$.


	\section{Preliminaries on computational complexity and safe recursion} \label{sec:prelims}\label{sec2}

Throughout this work we only consider (partial) functions on \emph{natural numbers}.
We write $|x|$ for the length of the binary representation of a number $x$, and 
for lists of arguments $\vec x=x_1, \dots, x_n$ we write $|\vec x|$ for the list $|x_1|, \dots, |x_n|$.

\subsection{Non-uniform complexity classes}
$\fptime$ is the class of (total) functions computable in polynomial time on a Turing machine.
The `non-uniform' class $\FPpoly$ is an extension of $\fptime$ that intuitively has access to a polynomial amount of `{advice}', determined only by the \emph{length} of the input.
Formally:

\begin{defn}
[Non-uniform polynomial time]
$\FPpoly$ is the class of functions $f(\vec x)$ for which there are strings $\alpha_{\vec n} \in \{0,1\}^*$, \Achange{of size polynomial in $\vec n$}, and some $f'(x,\vec x)\in \fptime$ with:
\begin{itemize}
    \item $|\alpha_{\vec n}| $ is polynomial in $\vec n$.
    \item $f(\vec x) = f'(\alpha_{|\vec x|}, \vec x)$.
\end{itemize}
\end{defn}
\Achange{The strings $\{\alpha_{\vec n}\}_{\vec n}$ represent the \emph{polynomial advice} given to a polynomial-time computation, here $f'(x,\vec x)$. 
Note that $f(\vec x)$ only `receives advice' depending on the lengths of its inputs, $\vec x$.}

Note, in particular, that $\FPpoly$ admits undecidable problems.
E.g.\ the function $f(x)=1$ just if $|x|$ is the code of a halting Turing machine (and $0$ otherwise) is in $\FPpoly$.
Indeed, the point of the class $\FPpoly$ is to rather characterise a more non-uniform notion of computation.
In particular, the following is well-known \Achange{(see, e.g., \cite[Theorem~6.11]{Arora-Barak})}:

\begin{prop}
\label{prop:fppoly-circuits}
$f(\vec x) \in \FPpoly$ iff there are polynomial-size circuits computing $f(\vec x)$.
\end{prop}

\subsection{The Bellantoni-Cook algebra}
A \emph{two-sorted} function is a function $f(\vec x;\vec y)$ whose arguments have been delimited into `normal' ones ($\vec x$, left of `;'), and `safe' ones ($\vec y$, right of `;').

The two-sorted algebra $\bc$ was introduced in \cite{BellantoniCook} and is defined as follows: 

\begin{defn}
[Bellantoni-Cook]
$\bc$ is the smallest class of two-sorted functions containing,
\begin{itemize}
    \item $0(;) := 0$
    \item $\succ 0(; x) := 2x$
    \item $\succ 1 (;x) := 2x+1$
    \item $\pi^{m;n}_{j;}(x_0, \dots, x_{m-1}; y_0, \dots, y_{n-1}) := x_j$, whenever $j<m$.
    \item $\pi^{m;n}_{;j}(x_0, \dots, x_{m-1}; y_0, \dots, y_{n-1}) := y_j$, whenever $j<n$.
    \item $\pred(;x) = \lfloor \frac x 2 \rfloor$ 
    \item $\cnd (;w,x,y,z) := \begin{cases}
    x & w=0 \\
    y & w= 0 \mod 2, w \neq 0 \\
    z & w = 1 \mod 2
    \end{cases}$
\end{itemize}
and closed under:
\begin{itemize}
    \item (Safe composition)
    \begin{itemize}
        \item if $g(\vec x;)\in \bc$ and $h(\vec x,x;\vec y)\in \bc$ then also $f(\vec x;\vec y) \in \bc$ where $f(\vec x;\vec y):= h(\vec x,g(\vec x;);\vec y)$.
        \item if $g(\vec x;\vec y) \in \bc$ and $h(\vec x; \vec y,y)\in \bc$ then also $f(\vec x;\vec y)\in \bc$ where $f(\vec x;\vec y):= h(\vec x;\vec y, g(\vec x;\vec y))$ 
    \end{itemize}
    \item (Safe recursion on notation)
    if $g(\vec x;\vec y)\in \bc$ and $h_0(x,\vec x;\vec y,y),h_1(x,\vec x;\vec y,y)\in \bc$ then also $f(x,\vec x;\vec y) \Achange{\in} \bc$ where:
    \[
    \begin{array}{r@{\ := \ }ll}
        f(0,\vec x;\vec y) & g(\vec x;\vec y) \\
        f(\succ 0 x, \vec x;\vec y) & h_0(x,\vec x;\vec y, f(x, \vec x;\vec y)) & x\neq 0 \\
        f(\succ 1 x, \vec x; \vec y) & h_1(x, \vec x; \vec y, f(x, \vec x;\vec y))
    \end{array}
    \]
\end{itemize}
\end{defn}

Safe composition ensures that safe arguments may never appear in a normal position.
Note that, in the recursion scheme, the recursion parameter is always a normal argument, whereas recursive calls must appear in safe position.
Along with the constraints on safe composition, this ensures that the position of a recursive call is never the recursion parameter of another recursion.
This seemingly modest constraint duly restricts computation to polynomial time, yielding Bellantoni and Cook's main result:
\begin{thm}
[\cite{BellantoniCook}]
$f(\vec x) \in \fptime$ if and only if $f(\vec x;) \in \bc$.
\end{thm}
\subsection{A proof-theoretic presentation of Bellantoni-Cook}

We shall work with a formulation of $\bc$ as a $S4$-style type system in sequent-calculus style, where modalities are used to distinguish the two sorts (similarly to \cite{Hofmann97}).

We consider \emph{types} (or \emph{formulas}) $\n$ (`safe') and $\sq \n$ (`normal') which intuitively vary over the natural numbers. 
We write $A,B,$ etc.\ to vary over types. A \emph{sequent} is an expression $\Gamma \seqar A$, where $\Gamma$ is a list of types (called the \emph{context} or \emph{antecedent}) and $A$ is a type (called the \emph{succedent}). For a list of types $\Gamma = \lists{k}{\n}{\n}$, we write $\sq\Gamma$ for $\lists{k}{\sn}{\sn}$. 

\begin{defn}
\label{defn:bc-derivations}
A \emph{$\bc$-derivation} is a (finite) derivation built from the rules in~\cref{fig:bc-type-system}. 
\end{defn}

\begin{figure}
$$
\small 
   \def\arraystretch{3}
   \begin{array}{c}
\vlinf{\id}{}{\n \seqar \n}{}
\quad 
\vliinf{\cut_\n}{}{\orange \Gamma \seqar B}{\orange \Gamma\seqar \n}{\orange \Gamma, \blue \n \seqar B}
\quad
\vliinf{\cut_\sq}{}{\orange \Gamma \seqar B}{\orange \Gamma\seqar \sn}{\blue{\sn},\orange \Gamma \seqar B}
\\
\vlinf{\wk_\n}{}{\orange\Gamma, \blue \n \seqar B}{\orange \Gamma \seqar B}
\quad
\vlinf{\wk_{\sq}}{}{\blue{\sn},\orange\Gamma \seqar B}{\orange \Gamma \seqar B}
\quad
\vlinf\exch{}{\orange \Gamma, \blue B,\red A, \purple{\Gamma'} \seqar C}{\orange \Gamma, \red A, \blue B, \purple{\Gamma'} \seqar C}
\quad 
\vlinf{\sql}{}{\blue \sn, \orange \Gamma \seqar A}{\orange \Gamma, \blue \n \seqar A}
\quad
\vlinf{\sqr}{}{\orange {\sq \Gamma} \seqar \sn}{ \orange {\sq \Gamma} \seqar \n}
\\
\vlinf{\zero}{}{\seqar \n}{}
\quad 
\vlinf{\unit}{}{\seqar \n}{}
\quad
\vlinf{\succ 0 }{}{\orange \Gamma \seqar A}{\orange \Gamma \seqar A}
\quad
\vlinf{\succ 1 }{}{\orange \Gamma \seqar A}{\orange \Gamma \seqar A}
\quad 
\vliiinf{\saferec}{}{\blue\sn, \orange \Gamma \seqar \n}{\orange \Gamma \seqar \n}{\blue\sn, \orange \Gamma  , \red \n \seqar \n}{\blue\sn, \orange \Gamma, \red \n \seqar \n}
\\ 
\vliiinf{\cnd_\n}{}{\orange \Gamma, \blue \n \seqar \n}{\orange \Gamma \seqar \n}{\orange \Gamma, \blue \n \seqar \n}{\orange \Gamma , \blue \n \seqar \n}
\qquad
\vliiinf{\cnd_\sq}{}{\blue{\sn},\orange \Gamma\seqar \n}{\orange \Gamma \seqar \n}{\blue{\sn},\orange \Gamma \seqar \n}{\blue{\sn}, \orange \Gamma  \seqar \n}
\\
\vliinf{\cndl_\n}{}{\orange \Gamma, \blue \n \seqar \n}{\orange \Gamma \seqar \n}{\orange \Gamma, \blue \n \seqar \n}
\qquad
\vliinf{\cndl_\sq}{}{\blue{\sn},\orange \Gamma\seqar \n}{\orange \Gamma \seqar \n}{\blue{\sn},\orange \Gamma \seqar \n}
\end{array}
$$
\caption{$\bc$ as a sequent-style type system.}
\label{fig:bc-type-system}
\end{figure}

The colouring of type occurrences in \cref{fig:bc-type-system} may be ignored for now, they will become relevant in the next section. 
We may write $\der= \rules (\der_1, \ldots, \der_n)$  
(for $ n\leq 3$)
if $\rules$ is the bottom-most inference step of a derivation $\der$ whose immediate subderivations are, respectively,  $\der_1, \ldots, \der_n$. As done in~\cite{CurziDas}, we shall assume w.l.o.g.~that sequents have shape ${{\sn}, \ldots, {\sn}, {\n}, \ldots, {\n} \seqar A}$, i.e.\ in the left-hand side all $\sn$ occurrences are placed before all $\n$ occurrences.

We construe the system of $\bc$-derivations as a class of two-sorted functions by identifying each rule instance as an operation on two-sorted functions as follows:

\begin{defn}
[Semantics of $\bc$]
\label{defn:semantics-bc}
Given a $\bc$-derivation $\der$ \Achange{of} $\lists{m}{\sn}{\sn}, \lists{n}{\n}{\n} \seqar A$ 
we define a two-sorted function $\denot \der (x_1, \dots, x_m;y_1, \dots, y_n)$ in~\Cref{fig:semantics-of-bc} by induction on the structure of $\der$ (all rules as typeset in \Cref{fig:bc-type-system}).
\end{defn}

\begin{figure}
\centering{
$$
    \small
\begin{array}{cc}
     \arraycolsep=1pt
        \def\arraystretch{1.2}
\begin{array}{rcl}
    \denot{\id} (;y)  & \dfn & y\\
  \denot {\cut_\n(\der_0, \der_1)} (\vec x;\vec y) &\dfn& \denot{\der_1}(\vec x; \vec y, \denot{\der_0} (\vec x; \vec y))\\
  \denot { \cut_\sq(\der_0, \der_1)} (\vec x;\vec y) & \dfn & \denot{\der_1}(\denot{\der_0} (\vec x; \vec y), \vec x; \vec y) \\
  \denot {\wk_\n(\der_0)} (\vec x; \vec y, y)   & \dfn  & \denot{\der_0} (\vec x;\vec y)\\
  \denot {\wk_\sq(\der_0)} (x, \vec x; \vec y) &\dfn& \denot{\der_0} (\vec x;\vec y)\\
  \denot {\exch_\n(\der_0)} (\vec x; \vec y, y,y', \vec y') &\dfn& \denot{\der_0} (\vec x; \vec y, y',y, \vec y')\\
  \denot {\exch_\sq(\der_0)} (\vec x, x, x', \vec x'; \vec y) & \dfn & \denot{\der_0} (\vec x, x',x , \vec x'; \vec y)\\
  \denot { \sql(\der_0)} (x, \vec x; \vec y) &\dfn& \denot{\der_0}(\vec x; \vec y, x)\\
  \denot {\sqr(\der_0)} (\vec x;)&  \dfn& \denot{\der_0} (\vec x;)\\
  \denot {i}(;) & \dfn & i   \\
   \denot {\succ i(\der_0)} (\vec x;\vec y) &\dfn& \succ i (; \denot{\der_0}(\vec x;\vec y)) 
\end{array}
     &  
     \arraycolsep=1pt
     \def\arraystretch{1.2}
     \begin{array}{rcl}
     \denot {\srec(\der_0, \der_1, \der_2)} (0, \vec x;\vec y) & \dfn &  \denot{\der_0} (\vec x;\vec y) \\
         \denot {\srec(\der_0, \der_1, \der_2)} (\succ i x, \vec x; \vec y) & \dfn & \denot{\der_{i+1}} (x, \vec x; \vec y, \\
         &&  \denot{\srec(\der_0, \der_1, \der_2)}(x, \vec x; \vec y))   \\ 
       \denot {\cnd_\n(\der_0, \der_1, \der_2)} (\vec x;\vec y,0) & \dfn &  \denot{\der_0} (\vec x;\vec y) \\
         \denot {\cnd_\n(\der_0, \der_1, \der_2)} (\vec x; \vec y, \succ i y) & \dfn & \denot{\der_{i+1}} (\vec x; \vec y, y)   \\ 
           \denot{\cnd_\sq(\der_0, \der_1, \der_2)} (0, \vec x;\vec y) & \dfn &  \denot{\der_0} (\vec x;\vec y) \\
         \denot{\cnd_\sq(\der_0, \der_1, \der_2)} (\succ i x, \vec x; \vec y) & \dfn & \denot{\der_{i+1}} (x, \vec x; \vec y)  \\ 
         \denot {\cndl_\n(\der_0, \der_1)} (\vec x;\vec y,0) & \dfn &  \denot{\der_0} (\vec x;\vec y) \\
         \denot {\cndl_\n(\der_0, \der_1)} (\vec x; \vec y, \succ i y) & \dfn & \denot{\der_1} (\vec x; \vec y, y) \\ 
           \denot{\cndl_\sq(\der_0, \der_1)} (0, \vec x;\vec y) & \dfn &  \denot{\der_0} (\vec x;\vec y) \\
         \denot{\cndl_\sq(\der_0, \der_1)} (\succ i x, \vec x; \vec y) & \dfn & \denot{\der_1} (x, \vec x; \vec y) 
     \end{array}
\end{array}
$$
}
\caption{Semantics of system $\bc$, where $i \in \{0,1\}$ and \Achange{$\succ i x \neq 0 $ and $\succ i y \neq 0$}. 
}   
\label{fig:semantics-of-bc}
\end{figure}

This formal semantics exposes how $\bc$-derivations and functions in the algebra $\bc$ relate. The rule $\saferec$ in~\Cref{fig:bc-type-system} corresponds to safe recursion, and safe composition along safe parameters is expressed by  $\cut_\n$. 
Note, however, that the interpretation of $\cut_\sq$ in~\Cref{fig:semantics-of-bc} apparently does not satifsfy the required constraint on safe composition of a function $g$ along a normal parameter of a function $h$, which forbids the presence of safe parameters in $g$. 
However, this admission turns out to
be harmless, and we are able to obtain the following result that justifies the overloading of the notation `$\bc$':
\begin{prop}[\cite{CurziDas}]
\label{prop:bc-type-system-characterisation}
 $f(\vec x;\vec y)\in \bc$ iff there is a $\bc$-derivation $\der$ for which $\denot \der (\vec x; \vec y) = f(\vec x;\vec y)$.
\end{prop}

\begin{rem}
[Bootstrapping]
\label{rem:extra-rules-redundant}
Note that the rules  $1$, $\cndl_\n$ and $\cndl_\sq$ are semantically redundant, being derivable from the others by: $\denot{1}= \denot{\succ 1(0)}$,  $ \denot {\cndl_\n(\der_0, \der_1)}=\denot {\cnd_\n(\der_0, \der_1, \der_1)}$, and $ \denot {\cndl_\sq(\der_0, \der_1)}=\denot {\cnd_\sq(\der_0, \der_1, \der_1)}$. 
Indeed, our original presentation of the system in \cite{CurziDas} did not include these rules, but
we have `bootstrapped' our system here in order to facilitate the definitions of our restricted `non-wellfounded' systems later for characterising $\FPpoly$, in particular in \cref{sec:cbc-nucbc}.
\end{rem}


\section{Non-wellfounded systems based on Bellantoni-Cook}
\label{sec:cbc-nucbc}
In this section we recall a `coinductive' version of $\bc$ that was recently introduced in our earlier work \cite{CurziDas}, and go on to introduce the new system $\nucbc$ of this work. 
In particular we shall give global  criteria that control the computational strength of non-wellfounded typing derivations. 
Throughout this section we shall work with the system $\bcnorec \dfn \bc \setminus \{\saferec\}$.


\begin{defn}
[Coderivations]
A ($\bcnorec$-)\emph{coderivation} $\der$ is a possibly infinite {rooted} tree generated by the rules of $\bcnorec$.
Formally, we identify $\der$ with a (labelled) prefix-closed subset of $ \{0,1,2\}^*$ (i.e.\ a ternary tree).
Each node is labelled by an inference step from $\bcnorec$ such that, whenever $\nu\in \der$ is labelled by a step $\vliiinf{}{}{S}{S_1}{\cdots}{S_n}$, for $n\leq 3$, $\nu$ has $n$ children in $\der$ labelled by steps with conclusions $S_1, \dots, S_n$ respectively. 
Sub-coderivations of a coderivation $\der$ rooted at position $\nu \in \{0,1,2 \}^*$ are denoted $\pder\nu$, so that $\pder\epsilon=\der$. 
\end{defn}

\begin{figure}[t]
    \centering
    $$
\footnotesize
\vlderivation{
\vliin{\cut_{\sq}}{\bullet}{\blue{\sn} \seqar \n}{
    \vlin{ \sqr}{}{\blue{\sn} \seqar \sn}{\vlin{\sql}{}{\blue{\sn} \seqar \n}{\vlin{ \succ 1}{}{\n \seqar \n}{\vlin{\id}{}{\n \seqar \n}{\vlhy{}}}}}
}{
   \vlin{\cut_{\sq}}{\bullet}{\sn \seqar \n}{\vlhy{\vdots}}
}
}
\qquad  
\vlderivation{
\vliin{\cnd_\sq}{\bullet\quad {\scriptstyle{i=0,1}}}{\blue{\underline{\sn}}, \sq \vec \n \seqar \n}
{
\vltr{\gfunction}{\sq \vec \n \seqar \n}{\vlhy{\ }}{\vlhy{\ }}{\vlhy{\ }}
}
{
  \vliin{\cut_\sq}{}{\blue{\sn}, \sq \vec \n \seqar \n}
  {
  \vlin{\sqr}{}{\blue{\sn}, \sq \vec \n  \seqar \red{\sn}}{
    \vlin{\cnd_{\sq}}{\bullet}{\blue{\sn}, \sq \vec \n  \seqar {\n}}{\vlhy{\vdots}}
  }
  }{
  \vltr{\hfunction_i}{\blue{\sn}, \sq \vec \n, \red{\sn}  \seqar \n}{\vlhy{\ }}{\vlhy{\ }}{\vlhy{\ }}
  }
}
{
}
} 
$$
$$
\footnotesize
\vlderivation{
\vliin{\cnd_\sq}{\bullet\quad {\scriptstyle{i=0,1}}}{\blue{\underline{\sn}}, \red{\sn} , \n \seqar \n }
 {
   \vliin{\cnd_\sq}{\circ\quad {\scriptstyle{i=0,1}}}{\red{\underline{\sn}}, \n \seqar \n }
   {
    \vlin{\id}{}{\n \seqar \n}{\vlhy{}}
   }
   {
   \vlin{\succ i}{}{\red{\sn},  \n \seqar \n}{\vlin{\cnd_\sq}{\circ}{ \red{\sn}, \n \seqar \n}{\vlhy{\vdots}}}
   }
 }
 {
  \vlin{\succ i}{}{\blue{\sn}, \red{\sn}, \n \seqar \n}{\vlin{\cnd_\sq}{\bullet}{\blue{\sn}, \red{\sn}, \n \seqar \n}{\vlhy{\vdots}}}
 }
}
$$
$$
\footnotesize
\vlderivation{
    \vliin{\cnd_{\sq}}{\bullet}{\blue{\underline{\sn}}, \orange \n \seqar \n}{
        \vlin{\succ{0}}{}{\orange \n \seqar \n}{
        \vlin{\id}{}{\orange \n \seqar \n}{\vlhy{}}
        }
    }{
        \vliin{\cut_\n}{}{\blue{\sn}, \orange \n \seqar \n}{
            \vlin{\cnd_{\sq}}{\bullet}{\blue{\sn}, \orange \n \seqar \red \n}{\vlhy{\vdots}}
        }{
            \vlin{\cnd_{\sq}}{\bullet}{\blue{\sn} , \red \n \seqar \n}{\vlhy{\vdots}}
        }
    }
}
\ \ 
\vlderivation{
 \vliin{\cndl_\sq}{}{\red{\underline{\sn}} \seqar \n}
  {
        \vltr{f(0)}{ \seqar \n}{\vlhy{\ }}{\vlhy{\  \ }}{\vlhy{\ }}
  }
   {
        \vliin{\cndl_\sq}{}{\red{\underline{\sn}} \seqar \n}
  {
        \vltr{f(1)}{ \seqar \n}{\vlhy{\ }}{\vlhy{\ \  }}{\vlhy{\ }}
  }
  { 
  \vlin{\cndl_\sq}{}{\red{\underline{\sn}} \seqar \n}{\vlhy{\vdots}}
    }
  }
}
\vspace{0.2cm}
$$
$$
\footnotesize
\vlderivation{
\vliin{\cndl_\sq}{\bullet}{\red{\underline{\sn}} \seqar \n}{\vlin{0}{}{\seqar \n}{\vlhy{}}}
   {
   \vliin{\cut_{\n}}{}{\red{\sn} \seqar \n}
      {
      \vlin{\succ 0}{}{\red{\sn} \seqar \blue{\n}}{\vlin{\cndl_\sq}{\bullet}{\red{\sn} \seqar \n}{\vlhy{\vdots}}}
      }{
      \vliin{\cut_{\n}}{}{\red{\sn}, \blue{\n} \seqar \n }
      {
      \vlin{\succ 1}{}{\red{\sn} \seqar \green{\n}}{\vlin{\cndl_\sq}{\bullet}{\red{\sn} \seqar \n}{\vlhy{\vdots}}}
      }{
      \vliin{\cut_{\n}}{}{\red{\sn}, \blue{\n} , \green{\n} \seqar \n}
      {
      \vltr{\completeness(r)}{\red{\sn} \seqar \orange{\n}}{\vlhy{\ }}{\vlhy{\ \ \ }}{\vlhy{\ }}
      }{
      \vliiin{\cnd_\n}{}{\orange{\n} , \blue{\n} , \green{\n} \seqar \n }{\vlin{\id}{}{\blue{\n} \seqar \n }{\vlhy{}}}{\vlin{\id}{}{\blue{\n} \seqar \n }{\vlhy{}}}{\vlin{\id}{}{\green{\n} \seqar \n }{\vlhy{}}}
      }
      }
      }
   }
}
$$
    
    \caption{Examples of coderivations:  $\incrementation$ (top left), $\primrec$ (top right), $ \cconc$ (second line), $\sexp$ (third line, left),   $\completeness(f)$ with  $f: \Nat \to \Nat$ (third line, right), 
    \Achange{$\mathcal{A}(r)$ with $r:\Nat \to \{0,1\}$ (bottom)}.
    }
  \label{fig:examples-regular-coderivations}
\end{figure}

Examples of coderivations can be found    in~\Cref{fig:examples-regular-coderivations} (some of them are  from~\cite{CurziDas}), {whose computational meaning is discussed in~\Cref{exmp:sem}},  and employ the following  conventions:

\begin{conv}[Representing coderivations]\label{conv:convention}
Henceforth, we may mark steps by $\bullet$ (or similar) in a coderivation to indicate roots of identical sub-coderivations. Moreover, to avoid ambiguities and  to ease parsing of (co)derivations, we shall often underline principal formulas of a  rule instance in a given coderivation and \Achange{omit instances of structural rules $\exch_\n$, $\exch_\sq$,  $\wk_\n$ and $\wk_\sq$, absorbing them into other steps (typically cuts) when it causes no confusion}. 
Finally, when the sub-coderivations $\der_0$ and $\der_1$ above the second and the third premise of the conditional rule (from left) are similar, we may compress them into a single `parametrised' sub-coderivation $\der_i$ (with $i=0,1$).
\end{conv}

As discussed in~\cite{Das2021,Das2021-preprint,Kuperberg-Pous21}, coderivations can be identified with Kleene-Herbrand-G\"odel style equational programs, in general computing partial recursive functionals  (see, e.g., \cite[\S 63]{Kleene71:intro-to-metamath} for further details).
We shall specialise this idea to our two-sorted setting.

\begin{defn}
[Semantics of coderivations]
\label{defn:semantics-coderivations}
To each $\bcnorec$-coderivation $\der$ we associate a two-sorted Kleene-Herbrand-G\"odel partial function $\denot \der$ obtained by construing the semantics of \Cref{defn:semantics-bc} as a (possibly infinite) equational program.
Given a two-sorted function $f(\vec x; \vec y)$, we say that $f$ is \emph{defined} by a $\bcnorec$-coderivation $\der$ if $\denot\der(\vec x; \vec y)=f(\vec x; \vec y)$. 
\end{defn}

\begin{rem}\label{rem:equational-programs}
The notion of \emph{computation} for equational programs is given by (finitary) reasoning in equational logic (see, e.g., \cite[\S 63]{Kleene71:intro-to-metamath}): for numerals $\vec m,\vec n$, we have that $\denot \der (\vec m;\vec n)$ is well-defined and returns some numeral $k$ just if the equation $\denot \der(\vec m;\vec n)=k $ can be (finitely) derived in equational logic (with basic numerical axioms) over the equational program for $\der$.
\Achange{Implicit here is the fact that the semantics of $\bcnorec$-coderivations yield \emph{coherent} equational programs: whenever $\denot \der (\vec m;\vec n)=k$ and $\denot\der (\vec m;\vec n)=k'$ are derivable then $k=k'$ \cite{Das2021,Das2021-preprint}.}    
\end{rem}

\begin{exmp}\label{exmp:sem}
By purely equational reasoning, we can simplify the Kleene-G\"{o}del-Herbrand style semantics in~\Cref{defn:semantics-coderivations} of the coderivations in~\Cref{fig:examples-regular-coderivations} to get the  equational programs in~\Cref{fig:examples-equational-programs}:  $\denot{\incrementation}$ represents a  function that is always undefined, as its equational program keeps increasing the length of the input; $\denot{\primrec}$ is an instance of a \emph{non-safe}  recursion scheme (on notation), as the recursive call appears in normal position; $\denot{\cconc}$ computes concatenation of the binary representation of three natural numbers;  $\denot{\sexp}$ has exponential growth rate (as long as $y\neq 0$), since  $\denot\sexp(x;y)=2^{2^{\s{x}}} \cdot  |y|$;  the (infinite) equational program \Achange{ for $\denot{\completeness(f)}$ computes $f(|x|)$ by simply exhausting the values of $|x|$;}
finally, \Achange{$\denot {\mathcal{A}(r)}$ on input $x$ returns the binary string  $r(0)\cdot r(1)\cdot  \cdots \cdot   r(|x|-1)$} if $x>0$, and  $0$ otherwise.
\end{exmp}

\begin{figure}
    \centering
$$
\small
\begin{array}{cc}
\arraycolsep=1pt
 \def\arraystretch{1.2}
    \begin{array}{rcl}
\denot{\incrementation}(x;)  &=&  \denot{\incrementation}(\succ 1x;) \\ 
\denot{\primrec}(0, \vec x;) & = & \denot{\gfunction}(\vec x; )\\
  \denot{\primrec}(\succ ix,\vec x; ) & = &\denot{\hfunction_i}( x, \vec x, \denot{\primrec}(x, \vec x; ); )   \\
  \denot{\cconc}  (0, 0;z) & = & z\\
     \denot{\cconc}  ( 0, \succ i y;z) & = &   \succ i \denot{\cconc} ( 0, y;z)    \\ 
     \denot{\cconc}  (\succ i x, y;z) & = &   \succ i \denot{\cconc}  ( x, y;z) 
\end{array}
&
\arraycolsep=1pt
 \def\arraystretch{1.2}
\begin{array}{rcl}
\denot \sexp (0;y) & = & \succ 0 (;y) \\
    \denot \sexp (\succ i x;y) & = & \denot \sexp (x; \denot \sexp(x;y)) 
    \\  
     \{ \denot{\completeness(f)}(x;) &= & \Achange{f(|x|)\}_{ |x| \in \Nat}  } \\
     \denot {\mathcal{A}(r)} (0; )   & = & 0 \\
    \Achange{\denot {\mathcal{A}(r)} (\succ i x; )} & = & \begin{cases}\succ 0 \denot {\mathcal{A}(\Achange r)} ( x; ) &\text{if } \Achange{\denot{\completeness(r)}(x;)=0} \\
    \succ 1 \denot{\mathcal{A}(\Achange r)} ( x; ) &\text{otherwise}\end{cases}
    
\end{array}
 \end{array}
$$
\caption{Equational programs derived from the coderivations in~\Cref{fig:examples-regular-coderivations}, where $i \in \{0,1\}$.}
\label{fig:examples-equational-programs}
\end{figure}


The above examples illustrate several recursion theoretic features of $\bcnorec$-coderivations that we shall seek to control in the remainder of this section:
\begin{enumerate}[(I)]
    \item \label{enum:problem1}  \emph{non-totality} (e.g., the coderivation $\incrementation$);
    \item \label{enum:problem2}  \emph{non-computability} (e.g., the \Achange{coderivation $\completeness(f)$, with $f$ non-computable}); 
    \item \label{enum:problem3}  \emph{non-safety} (e.g., the coderivation $\primrec$), despite the presence of modalities implementing the normal/safe distinction of function arguments;
    \item \label{enum_problem4} \emph{nested recursion} (e.g., the coderivation $\sexp$).
\end{enumerate}


To address \eqref{enum:problem1} we shall adapt to our setting a well-known `totality criterion' from non-wellfounded proof theory (similar to those in~\cite{Das2021,Das2021-preprint,Kuperberg-Pous21}).
First we need to recall some standard structural proof theoretic notions:

\begin{defn}
[Ancestry]
\label{defn:ancestry}
Fix a coderivation $\der$. We say that a type occurrence $A$ is an \emph{immediate ancestor} of a type occurrence $B$ in $\der$ if they are types in a premiss and conclusion (respectively) of an inference step and, as typeset in \Cref{fig:bc-type-system}, have the same colour.
If $A$ and $B$ are in some $\orange \Gamma$ or $\purple{ \Gamma'}$, then furthermore they must be in the same position in the list.
\end{defn}

\Achange{For a definition of immediate ancestry avoiding colours, we point the reader to standard proof theory references, e.g.\ \cite[Sec.~1.2.3]{buss98:intro-in-handbook}.}
Being a binary relation, immediate ancestry forms a directed graph upon which our totality criterion is built:

\begin{defn}
[Progressing coderivations]
\label{defn:progressing}
Fix a coderivation $\der$. A \emph{thread} is a maximal path in the graph of immediate ancestry.
We say that a (infinite) thread is \emph{progressing} if it is eventually constant $\sn$ and infinitely often principal for a $\cnd_{\sq}$ rule or a $\cndl_\sq$ rule. A coderivation is \emph{progressing} if each of its infinite branches has a progressing thread.
\end{defn}

In~\cite{CurziDas} we showed that the progressing criterion is indeed sufficient (but obviously not necessary) to guarantee that the partial function computed by a coderivation is, in fact, total  (see also \cite{Kuperberg-Pous21,Das2021-preprint,Das2021}):

\begin{prop}
[Progressing implies totality, \cite{CurziDas}]
\label{prop:prog-imp-tot}
If $\der $ is progressing, then $\denot\der$ is total.
\end{prop}
The argument for this proposition is by contradiction: assuming non-totality, construct an infinite `non-total branch', whence a contradiction to well-orderedness of $\Nat$ is implied by a progressing thread along it.
We shall use similar argument later in the proof of  \cref{lem:succ-id-free-imp-relation}.

\begin{exmp}\label{exmp:1} In~\Cref{fig:examples-regular-coderivations},  $\incrementation$ has precisely one infinite branch (that loops on $\bullet$) which contains no instances of  $\cnd_{\sq}$ or $\cndl_\sq$ at all, so  $\incrementation$ is not progressing. 
On the other hand, $\cconc$ has two simple loops, one on $\bullet$ and the other one on $\circ$. For any infinite branch $B$ we have two cases: if $B$ crosses the bottommost conditional infinitely many times, it contains a progressing \blue{blue} thread; otherwise, $B$ crosses the topmost conditional infinitely many times, so that it  contains  a progressing \red{red} thread.
Therefore, $\cconc$ is progressing. By applying the same reasoning, we conclude that $\sexp$, $\completeness(f)$, $\mathcal{A}(r)$, and   $\primrec$ are progressing (if $\gfunction$ and $\hfunction_i$ are).
\end{exmp}

To address \eqref{enum:problem3}-\eqref{enum_problem4} we recall the following properties of coderivations from \cite{CurziDas}:

\begin{defn}[Safety, left-leaning]\label{defn:conditions} We say that a coderivation $\der$ is  \emph{safe} if each  branch crosses only finitely many $\cut_\sq$-steps, and  \emph{left-leaning} if each  branch goes right at a $\cut_\n$-step only finitely often.
\end{defn}

\begin{exmp}\label{exmp:2}
In~\Cref{fig:examples-regular-coderivations}, the only non-safe coderivations are $\primrec$ and $\incrementation$, as the  branches looping on $\bullet$  contain infinitely many $\cut_{\sq}$.     $\sexp$ is the only non-left-leaning coderivation, as it has a branch looping at $\bullet$ that crosses  infinitely many times the rightmost premise of a  $\cut_\n$. 
\end{exmp}

Finally, concerning \eqref{enum:problem2}, recall that the aim of this work is to characterise non-uniform classes, which may contain non-computable predicates and functions. 
To this end we introduce a generalisation of the notion of `regularity', typically corresponding to computability (e.g.\ in \cite{Das2021-preprint,Das2021,Kuperberg-Pous21}), that is commonplace in cyclic proof theory:

\begin{defn}
[Generalised regularity]
\label{defn:regularity}
\Achange{Let $\mathsf R\subseteq \bcnorec$.
A $\bcnorec$-coderivation $\der$ is \emph{$\mathsf R$-regular} if it has only finitely many distinct sub-coderivations containing rules among $\mathsf R$. 
If $\mathsf R = \bcnorec$,} i.e.~it has only finitely many distinct sub-coderivations, then we say that $\der$ is \emph{regular} \Achange{(or \emph{circular})}. 
\end{defn}

Note that, while usual derivations may be naturally written as finite trees or dags, regular coderivations may be naturally written as finite directed (possibly cyclic) graphs. Also, from a regular coderivation $\der$ we obtain a \emph{finite} equational program for $\denot \der$. 
In particular, while there are continuum many (non-wellfounded) coderivations, there are only countably many regular ones.

\begin{exmp}\label{exmp:3}
In~\Cref{fig:examples-regular-coderivations},  $\completeness(f)$ and \Achange{$\mathcal{A}(r)$
are the only non-regular coderivations (as long as $\gfunction$, $\hfunction_i$ are regular).
Also, $\mathcal{A}(r)$ is \Achange{$\mathsf R$-regular for any $\mathsf R\subseteq \bcnorec\setminus\{0, 1, \cndl_\sq\}$}, since $r(i)$ is computed by just a $0$ or $1$ step when $r:\Nat \to \{0,1\}$.}
\end{exmp}

We are now ready to present the non-wellfounded proof systems that will be considered in this paper:

\begin{defn}[$\cbc$ and $\nucbc$]\label{defn:non-wellfounded-proof-systems}
$\cbc$ is the class of regular progressing safe and left-leaning $\bcnorec$-coderivations. $\nucbc$ is the class of $\relaxed$-regular progressing safe and left-leaning $\bcnorec$-coderivations. A two-sorted function $f(\vec x;\vec y)$ is  \emph{$\cbc$-definable} (resp.~\emph{$\nucbc$-definable}) if there is a coderivation $\der \in \cbc$  (resp.~$\der \in \nucbc$ ) such that $\denot \der (\vec x;\vec y) = f(\vec x;\vec y)$. 
\end{defn}

Recalling Examples~\ref{exmp:1},~\ref{exmp:2} and~\ref{exmp:3}, $\cconc$ is the only coderivation in $\cbc$, \Achange{while $\mathcal{A}(r)$} is an example of coderivation in $\nucbc$ for any \Achange{$r:\Nat \to \{0,1\}$}. The system $\cbc$ was already introduced
in~\cite{CurziDas}, where we showed that $\cbc= \fptime$ (among other results), whereas $\nucbc$ \Achange{(read `non-uniform $\bc$')} is new. 
The main result of this paper  
is to show that $\nucbc$ admits just the right amount of non-wellfoundedness to duly characterise the analogous non-uniform class:

\begin{thm}\label{thm:main-result}
$\nucbc = \FPpoly$
\end{thm}

\Achange{The rest of this work is devoted to the proof of this result. In particular, the two directions of the equality are given by~\Cref{cor:fppoly-in-nuB}  and~\Cref{cor:nuB-in-fppoly}.}


\begin{rem}
[On proof checking]
Let us point out that all conditions on coderivations we have considered so far are \emph{decidable} on regular coderivations.
In particular, progressiveness may be decided by reduction to universality of B\"uchi automata.
In the presence of safety, however, it turns out that proof  checking becomes easier: checking whether a regular coderivation is in $\cbc$ is actually decidable in $\mathbf{NL}$ \cite[Cor.~32]{CurziDas}.
Of course, as $\nucbc$ coderivations are not finitely presented (indeed like $\FPpoly$ programs), such decidability issues are no longer relevant.
\end{rem}

	\section{On relativised characterisations of $\FPpoly$}\label{sec3}\label{sec:rel-chars-fppoly}
In this section we consider recursion theoretic characterisations of $\FPpoly$ via relativised function algebras.
This will serve not only as a `warm up' to motivate our main characterisation, but will also provide several of the intermediate results necessary to that end.
The results of this section are based on textbook techniques and are (presumably) folklore.

\subsection{Non-uniformity via resource-bounded oracle machines}
A \emph{relation} is a function $r(\vec x)$ such that we always have $r(\vec x) \in \{0,1\}$.
\begin{definition}
[Relativised complexity classes]
Let $R$ be a set of relations.
The class $\fptime(R)$ consists of just the functions computable in polynomial time by a Turing machine with access to \Achange{an oracle for each $r \in R$}.
\end{definition}

For instance, using this notion of relativised computation, we can define the levels of the functional polynomial hierarchy $\mathbf{FPH}$ by $\Box^p_1:= \fptime$, $\Box^p_2 := \fptime(\mathbf{NP})$, $\Box^p_3:= \fptime (\Sigma^p_2)$, etc.

Let us write $\RR := \{ r: \mathbb{N}^k \to \Achange{\{0,1\}} \ \vert \   |\vec x|=|\vec y| \implies r(\vec x)=r(\vec y)  \}$.
Note that the notation $\RR$ is suggestive here, since its elements are essentially maps from lengths/positions to Booleans, and so may be identified with Boolean streams.

\begin{prop}
\label{prop:fppoly=fptime(RR)}
$\FPpoly = \fptime(\RR)$.
\end{prop}
\begin{proof}[Proof sketch]
\todo{clean this up a bit}
For the left-right inclusion, let \Achange{$p(n)$ be a polynomial and} $\mathbf C = (C_n)_{n<\omega}$ be a circuit family with each $C_n$ taking $n$ Boolean inputs and having size $<p(n)$. We need to show that the language computed by $\mathbf C$ is also computed in $\fptime(\RR)$.
Let $c \in \RR$ be the function that, on inputs $x,y$ returns the $|y|$\textsuperscript{th} bit of $C_{|x|}$.
Using this oracle we can compute $C_{|x|}$ by polynomially queries to $c$, and this may be evaluated as usual using a polynomial-time evaluator in $\fptime$.

For the right-left inclusion, notice that a polynomial-time machine can only make polynomially many calls to oracles with inputs of only polynomial size. Thus, if $f \in \fptime (\RR)$ then there is some $p_f$ with $f\in \fptime (\RR^{<p_f})$, where $\RR^{<p_f}$ is the restriction of each $r \in \RR$ to only its first $p_f(|\vec x|)$ many bits.
Now, since $f$ can only call a fixed number of oracles from $\RR$, we can collect these finitely many polynomial-length prefixes into a single advice string for computation in $\FPpoly$.
\end{proof}

\subsection{A relativised Bellantoni-Cook characterisation of $\FPpoly$}

\newcommand{\succi}[1]{\succ{}{#1}}

    
    
    
    

We shall employ the following writing conventions for the remainder of this work.
For a set of (single-sorted) functions $F$, let us write:
\begin{itemize}
    \item $\nor F$ for the set of two-sorted functions $f(\vec x;)$ for each $f(\vec x) \in F$;
    \item $\safnor F$ for the set of two-sorted functions $f(\vec x;\vec y)$ for each $f(\vec x,\vec y)\in F$.
\end{itemize}

Given a set $F$ of two-sorted functions, the algebra $\bc(F)$ is defined just like $\bc$ but with additional initial (two-sorted) functions $F$. 
\Achange{Note that, since functions of $\bc(F)$ are given by finite programs, they can only depend on finitely many members of $F$.}



\newcommand{\Eval}{\mathrm{Eval}}
\begin{prop}\label{prop:fppoly-in-B(R0)}
    $ \FPpoly \subseteq \bc(\RRnor)$
\end{prop}
One natural way to prove this result would be to go via $\fptime(\RR)$, in light of \cref{prop:fppoly=fptime(RR)}. 
Indeed Bellantoni established foundational results relating $\fptime(R)$ and versions of $\bc(R)$, for $R$ a set of relations, in \cite{Bellantoni-fph}, but unfortunately the sorting of the corresponding arguments is subtle and does not immediately give the result we are after. 
For this reason we give a direct proof, that nonetheless inlines some ideas from \cite{Bellantoni-fph}.
\begin{proof}
[Proof of \cref{prop:fppoly-in-B(R0)}]
    Let $\mathbf C = (C_n)_{n<\omega}$ be a circuit family with each $C_n$ taking $n$ inputs and having size $<p(n)$, for some (monotone) polynomial $p$.
    We need to show that the language computed by $\mathbf C$ is also computed in $\B(\RRnor)$.
    
    First, let $\Eval(x,y)$ evaluate the circuit described by $x$ on the input $y$.
    Since $\Eval \in \fptime$, we have as standard (e.g.\ by \cite[Lemma 3.2]{BellantoniCook}) a function $\Eval(m;x,y)\in \B$ and a monotone polynomial $q$ such that $|m|\geq q(|x|,|y|) \implies \Eval(m;x,y) = \Eval(x,y)$.
    Now, in particular, if $x$ is the description of some $C_n$ and $n=|y|$, then also $|x|\leq p(|y|)$, and so $|m|\geq q(p(|y|),|y|) \implies \Eval(m;x,y) = \Eval(x,y)$.
    Finally, {denoting  $\overset{n}{\overbrace{\succ 1\ldots \succ 1 \phantom{|}}}0$ \Achange{by} $1^n$}, this means that we have $\Eval(y;x) := \Eval(1^{q(p(|y|),|y|)};x,y) \in \B$, that in particular evaluates, when $x$ describes $C_{|y|}$, the circuit $ C_{|y|}$ on input $y$.
    
    Now, let $c\in \RRnor$ with $c(y,z;)=$  $|z|$\textsuperscript{th} bit of $C_{|y|}$.\anupam{read the description left-right, with infinitely many trailing zeroes. We assume that circuits are coded in a prefix-free way that the evaluator takes into consideration.}
    We show that the function $C(y,z;) = c(y,0;)\cdot c(y,1;)\cdot \cdots \cdot c(y,1^{|z|-1};)$ is in $\B(c)$ by the following instance of safe recursion:
    $$
    \begin{array}{r@{\ = \ }l}
         C(y,0;) & 0 \\
         C(y,\succ i z;) & \cnd (; c(y,z;), \succ 0 (;C(y,z;)), \succ 1 (; C(y,z;)))
    \end{array}
    $$
    So we have that $C(y;) := C(y,1^{p(|y|)};)$ computes the description of $C_{|y|}$.\anupam{...up to some trailing 0s that the evaluator ignores}
    Now we can decide whether $y$ is accepted by $C_{|y|}$ simply by calling the function $\Eval(y;C(y;)) \in \B(c)$.
\end{proof}




It turns out that we also have the converse inclusion too.
This will be subsumed by our later results but we include it here for the sake of completeness.
The key is to establish a general form of Bellantoni and Cook's polymax bounding lemma to account for modulus of continuity as well as growth:

\begin{lem}
[Relational bounding lemma for $\bc$]
Let $R$ be a set of two-sorted relations, and suppose $f(R)(\vec x;\vec y) \in \bc(R)$.
Then there is a polynomial $p_f$ such that, setting $m_f(\vec m,\vec n) := p_f(\vec m) +\max \vec n$, we have:
\begin{itemize}
    \item (Polynomial modulus of growth) $|f(R)(\vec x;\vec y)| < m_f(|\vec x|,|\vec y|)$
    \item (Polynomial modulus of continuity)
    $f(R)(\vec x;\vec y) = f(\lambda |\vec u|,|\vec v|<m_f(|\vec x|,|\vec y|) . r(\vec u;\vec v))_{r\in R}(\vec x;\vec y)$
\end{itemize}
\end{lem}

\noindent
Using this we may establish:
\begin{prop}
[E.g.\ see \cite{Bellantoni-fph}]
Let $R$ be a set of relations.
$\bc(\safnor R)\subseteq \fptime(R)$.
\end{prop}
We shall not actually need this result directly in this work, rather recovering (a version of) it from a more refined grand tour of inclusions.
However this does lead to the first `implicit' characterisation of $\FPpoly$ of this work:
\begin{cor}
$\bc(\RRnor) = \FPpoly$
\end{cor}

	\section{$\FPpoly\subseteq \nucbc$ via relativised circular systems}\label{sec4}\label{sec:fppoly-in-nuB}

In this section we establish one direction of \Cref{thm:main-result}.
In particular, by the end of this section, we will have established the following inclusions,
\[
\FPpoly \subseteq \bc(\RRnor)\subseteq \cbc(\RRnor)\subseteq \nucbc
\]
where \Achange{$\cbc(F)$} is an extension of $\cbc $ by new initial sequents for two-sorted functions \Achange{in $ F$}.

\subsection{Relativised simulation of $\bc$ in $\cbc$}
We shall consider `relativised' versions of $\cbc$, that may include new initial sequents.
Formally:

\begin{defn} \label{defn:oracles}
Let \Achange{$F$ be a set} of two-sorted functions. A $\bcnorec(F)$-coderivation is just a usual $\bcnorec$-coderivation that may use initial sequents of the form
$\vlinf{f}{}{\sn^{n_i},\n^{m_i}\seqar \n }{}$, when $f\in F $ takes $n_i$ normal and $m_i$ safe inputs. 
We write $ \cbc(F)$ for the set of $\cbc$-coderivations allowing initial sequents from $F$. 
%
The semantics of such coderivations and the notion of $\cbc(F)$-definability are as expected, 
with 
$\denot{\der(F)}$ denoting the induced interpretation of $\der(F) \in \cbc(F)$. 
\end{defn}

\Achange{Note, again, that since $\cbc(F)$ coderivations are regular, they only depend on finitely many members of $F$.}
By a modular extension of the result that $\bc \subseteq \cbc$ from \cite{CurziDas}, we obtain:
\begin{prop}
\label{prop:b(f)-in-cbc(f)}
Let $F$ be a set of two-sorted functions. $\bc(F)\subseteq \cbc(F)$.
\end{prop}
The proof is simply by structural induction on the definition of a $\bc(F)$ function, where the recursion cases are handled by circularity as in \cite{CurziDas}. In particular, if $f$ is defined by safe recursion on notation from $g,h_0,h_1$ then the corresponding $\cbc$-coderivation is given by:
$$
\small
\vlderivation{
\vliin{\cnd_\sq}{\bullet \quad {\text{\scriptsize$i=0,1$}}}{\sn, \Gamma \seqar N}{
    \vltr{g}{\Gamma \seqar N}{\vlhy{\ }}{\vlhy{}}{\vlhy{\ }}
}{
    \vliin{\cut_\n}{}{\sn, \Gamma \seqar N}{
        \vlin{\cnd_\sq}{\bullet}{\sn, \Gamma \seqar N}{\vlhy{\vdots}}
    }{
        \vltr{h_i}{\sn, \Gamma, N\seqar N}{\vlhy{\ }}{\vlhy{ }}{\vlhy{\ }}
    }
}
}
$$
The only new cases in the induction are for an initial function from $F$, which is simply translated into the appropriate initial sequent.



\subsection{Simulating $\RRnor$ oracles in $\nucbc$}
In this subsection we shall establish:
\begin{prop}\label{prop:cbc(R0)-to-nuB} 
$\cbc(\RRnor)\subseteq \nucbc$.
\end{prop}
By definition of $\nucbc$ and $\cbc$, it suffices to only consider the new initial sequents from $\RRnor$. For this we simply appeal to the following lemma:

\begin{lem}
\label{lem:ext-compl-rel-nor}
For each $r(\vec x;) \in \RRnor$, there is a $\nucbc$-coderivation defining it, in particular using only the rules $0,1,\cndl_\sq$.
\end{lem}
\begin{proof}
We proceed by induction on the length of $\vec x$. When the list is empty, then $r(;)$ is just a Boolean, in which case we can derive it with just the $0$ or $1$ step.
Now, for $r(x,\vec x;)$ we have:
$$
\small
\vlderivation{
\vliin{\cndl_\sq}{}{\underline\sn, \sq \vec N \seqar N }{
    \vltr{\text{IH($r_0$)}}{\sq \vec N \seqar N}{\vlhy{\quad }}{\vlhy{\quad}}{\vlhy{\quad }}
}{
    \vliin{\cndl_\sq}{}{\underline\sn , \sq \vec N \seqar N}{
    \vltr{\text{IH($r_1$)}}{\sq \vec N \seqar N}{\vlhy{\quad }}{\vlhy{\quad}}{\vlhy{\quad }}
}{\vlhy{\quad \vdots\quad }}
}
}
$$
where $r_i(\vec x)$ is the function $r(1^i,\vec x)$ and the coderivations marked IH($r_i$) are obtained by the inductive hypothesis for $r_i$.
\end{proof}

Note that, by putting together \cref{prop:fppoly-in-B(R0)}, \cref{prop:b(f)-in-cbc(f)} (setting $F = \RRnor$) and \cref{prop:cbc(R0)-to-nuB}, we now have one half of our main result:
\begin{cor}
\label{cor:fppoly-in-nuB}
$\FPpoly\subseteq \nucbc$
\end{cor}

\section{$\nucbc$ as relativised regular coderivations}\label{sec:nuB-in-cb(R)}
To facilitate the other direction of \cref{thm:main-result}, let us first address a form of converse to \cref{prop:cbc(R0)-to-nuB} above, that duly embeds $\nucbc$ into a relativised circular system, which we shall rely on in the next section:

\begin{thm}\label{lem:non-wellfounded-oracles} $\nucbc \subseteq \cbc(\RRsafnor)$.
\end{thm}

Before giving the proof, we need to first establish some intermediate results:

\begin{lem}
\label{lem:succ-id-free-imp-relation}
If $\der $ is progressing and $\{\succ 0 , \succ 1 , \id\}$-free then $\denot\der$ is a relation, i.e.\ $\denot \der (\vec x;\vec y) \leq 1$.
\end{lem}
\begin{proof}
[Proof sketch]
We proceed by contradiction, always assuming \cref{prop:prog-imp-tot}, that progressing coderivations compute total functions.

If $\denot\der (\vec x;\vec y)>1$ then we argue that there is an immediate sub-coderivation $\der'$ of $\der$ and arguments $\vec x',\vec y'$ such that $\denot{\der'}(\vec x';\vec y')>1$.
Some of the critical cases are:
\begin{itemize}
    \item If $\der=\cut_\n(\der_0, \der_1)$ then $\denot \der(\vec x; \vec y) =\denot {\der_1}(\vec x;\vec y, \denot {\der_0}(\vec x; \vec y))$, and  we set $\der'\dfn \der_1$,   $\vec x'\dfn \vec x$ and $\vec y'\dfn \vec y, \denot {\der_0}(\vec x; \vec y)$ (since  $\denot {\der_0}(\vec x; \vec y)$ is well-defined). The case for  $\cut_\sq$ is similar.
    \item If $\der=\cnd_\sq(\der_0, \der_1, \der_2)$ then $\vec x=x, \vec z$ with  $\denot \der(0, \vec z; \vec y)= \der_0(\vec z; \vec y)$ and  $\denot \der (\succ i x', \vec z; \vec y)=\denot {\der_{i+1}}(x', \vec z; \vec y)$. If $x=0$ we set $\der' \dfn \der_0$, $\vec x'\dfn \vec z$, and  $\vec y'\dfn \vec y$. If $x= \succ i x'$ then we set $\der' \dfn \der_{i+1}$, $\vec x'\dfn x', \vec z$, and  $\vec y'\dfn \vec y$. The cases for $\cnd_N, \cndl_N$, and $\cndl_\sq$ are similar.
    \item In all other cases $\der$ ends with a unary rule so that $\der'$ is the only immediate sub-coderivation, and $\vec x',\vec y'$ are determined by the semantics of the rule (cf.~\cref{fig:semantics-of-bc}).
\end{itemize}
Note that, in the absence of $\succ 0,\succ 1$, we indeed have that $f'(\vec x';\vec y') = f(\vec x;\vec y)>1$ so we may continually apply this process to build up a branch $ B = ({\der=\der^{(0)}}, \der^{(1)}, \der^{(2)}, \dots)$ and arguments ${(\vec x;\vec y)=(\vec x^{(0)};\vec y^{(0)})}, (\vec x^{(1)};\vec y^{(1)}), (\vec x^{(2)}; \vec y^{(2)}), \dots$ such that $\denot{\der^{(k)}}(\vec x^{(k)};\vec y^{(k)}) = \denot\der(\vec x;\vec y)>1$. Observe that $ B$ cannot end at an $\id$ step, by assumption that $\der$ is $\id$-free. Also, if  $ B$ ends at a $0$ or $1$ step  we have by construction that $\denot \der(\vec x;\vec y) \in \{0,1\}$, a contradiction. 
Thus $ B$ must be infinite. 
Since $\der$ is progressing there is a  progressing thread along $B$, \Achange{say} $(\sn^i)_{i\geq k}$, where each $\sn^i$ is an occurrence of $\sn$.
\Achange{Let us}
examine the values, say $x^i$, assigned to each $\sn^i$. Notice that:
\begin{itemize}
    \item by  inspection of the rules and their interpretations from \Cref{defn:semantics-coderivations}, we have that $x^{i+1}\leq x^i$; and,
    \item if $\sn^i$ is principal for a $\cnd_\sq$ or a $\cndl_\sq$ step then $x^{i+1}< x^i$.
\end{itemize}
If follows that $(x^i)_{i\geq k}$ is a non-increasing sequence of natural numbers that does not converge, contradicting the well-ordering property of $\Nat$.
\end{proof}

\begin{lem}
 If $\der$ is  $\{ \cnd_\sq,\cnd_N, \id \}$-free,  {$|\vec x| = |\vec x’|$ and $|\vec y|=|\vec y’| $},  then $\denot \der(\vec x;\vec y) = \denot \der(\vec x';\vec y')$, whenever $f(\vec x;\vec y)$ is well-defined.
\end{lem}
\begin{proof}
[Proof sketch]
Being given by an equational program, we have that $\denot \der (\vec x;\vec y) = m$ has a (finite) equational derivation for some $m\in \Nat$, by assumption that it is well-defined (cf.~\Cref{rem:equational-programs}).
Replacing $\vec x,\vec y$ by $\vec x',\vec y'$ in this derivation yields  $\denot\der(\vec x';\vec y') = m$ too. The only critical cases are the steps $\cndl_N $,  $\cndl_\sq$, whose semantics only depend on the length of their arguments. 
\end{proof}
Putting the two above Lemmata together we have:
\begin{prop}
\label{prop:semi-regular-imp-RRsafnor}
If $\der $ is progressing and $\{\cnd_\sq,\cnd_N,\succ 0 , \succ 1, \id\}$-free, then $\denot\der \in \RRsafnor$.
\end{prop}

Now we can prove \cref{lem:non-wellfounded-oracles}:
\begin{proof}
[Proof sketch]
Let $\der$ be a $\nucbc$-coderivation and let $V$ be the set of minimal nodes $\nu$ such that $\der_\nu$ is $\{\cnd_\sq, \cnd_N,\succ 0 , \succ 1 , \id\}$-free, and so by \cref{prop:semi-regular-imp-RRsafnor} we have that each $\denot{\der_\nu} \in \RRsafnor$.

Now, let $\der^V$ be obtained from $\der$ by simply deleting each sub-coderivation $\der_\nu$, for $\nu \in V$, and construing each of their conclusions as new initial sequents.
By definition of $\nucbc$, note that $\der^V$ is now a coderivation in $\cbc(\denot{\der_\nu})_{\nu \in V}\subseteq \cbc(\RRsafnor)$ and we are done.
\end{proof}

\section{$\nucbc\subseteq \FPpoly$: a relativised algebra subsuming circular typing}
\label{sec:cb(R)-in-fppoly}
The final part of our chain of inclusions requires us to translate (relativised) circular coderivations into an appropriate function algebra.
The idea is that, in the presence of safety, one can reduce circularity to a form of recursion on `permutations of prefixes' that nonetheless remains feasible.
This was (one of) the main result(s) of \cite{CurziDas} and, fortunately, we are able to import those results accounting only for additional initial relations.

\subsection{Safe recursion on permutations of prefixes}
Let us write $\vec x \permprefeq \vec y$ if $\vec x$ is a permutation of prefixes of $\vec y$, i.e.\ $\vec x = x_0, \dots, x_{n-1}$ and $\vec y = y_0 , \dots , y_{n-1}$ and there is a permutation $\pi : [n] \to [n]$ s.t.\ each $x_i$ is a prefix of $y_{\pi i}$.
We shall write $\vec x \permpref \vec y$ if for at least one $i <n$ we have that $x_i $ is a strict prefix of $y_{\pi i}$. Note in particular that $\permpref$ is a well-founded pre-order so admits an induction and recursion principles.

To formulate recursion over well-founded relations it is convenient to employ (two-sorted) oracles as placeholders for recursive calls.
Due to the necessary constraints on composition, we shall formally distinguish these oracles (metavariables, $a,b,$ etc.) from additional initial functions (metavariables $f,g$ etc.\ until now).

\begin{defn}\label{defn:bsubseteq}
\Achange{Let $F$ be a set of two-sorted functions.}
The algebra \Achange{$\bcpp( F, \vec a)$} is the smallest class of two-sorted functions containing,
\begin{itemize}
    \item all the initial functions of $\bc$;
    \item each \Achange{(two-sorted)} function $a_i $ \Achange{among $ \vec a$};
    \item each \Achange{(two-sorted)} function \Achange{$f \in F$};
\end{itemize}
and closed under:
\begin{itemize}
\item (Relativised safe composition)
\begin{itemize}
    \item \Achange{if $g(\vec x;\vec y) \in \bcpp(F,\vec a)   $ and $h(\vec x;\vec y, y) \in  \bcpp(F,\vec a)   $ then  ${f(\vec x;\vec y) := h(\vec x;\vec y, g(\vec x;\vec y))\in \bcpp(F,\vec a)}$};
    \item \Achange{if $g(\vec x;) \in  \bcpp(F)$ and $h(\vec x,x;\vec y) \in \bcpp(F,\vec a)  $ then $  { f(\vec x;\vec y) := h(\vec x, g(\vec x;);\vec y)\in \bcpp(F,\vec a)}$};
\end{itemize}
    \item
    (Safe recursion on $\permpref$)
    \newline
    \Achange{if $h(a)(\vec x;\vec y) \in \bcpp(F,a, \vec a)$ then ${ f(\vec x;\vec y) := h(\lambda \vec u\permpref \vec x, \lambda \vec v \permprefeq \vec y . f(\vec u;\vec v))(\vec x;\vec y)\in \bcpp(F,\vec a) }$}.
\end{itemize}
\end{defn}
To be clear, the `guarded' abstraction notation above is formally defined as $$
\left(\lambda \vec u\permpref \vec x, \lambda \vec v\permprefeq \vec y . f(\vec u;\vec v)\right) (\vec u';\vec v')
\ := \ 
\begin{cases}
f(\vec u';\vec v') & \vec u' \permpref \vec x, \vec v' \permprefeq \vec y \\
0 & \text{otherwise}
\end{cases}
$$

\Achange{$\bcpp(\emptyset,\vec a)$ is the same as the notion $\bcpp(\vec a)$ from \cite{CurziDas}}.
Note in particular the distinction between \Achange{$F$} and $\vec a$ in the safe composition scheme: when composing along a normal parameter (second line), the function $g(\vec x;)$ must not contain any oracles \Achange{among} $\vec a$.


Adapting the \emph{Bounding Lemma} from \cite[Lemma 38]{CurziDas} to account for further initial relations gives:
\begin{lem}
 [Relational Bounding lemma]\label{lem:rel-bou-lem}
 \Achange{Let $R$ be a set of two-sorted relations and $f(R)(\vec x;\vec y) \in \bcpp (R)$.} 
 There is a polynomial $p_f(\vec n)$ such that, writing $m_f(\vec x,\vec y) = p_f(|\vec x|) + \max |\vec y|$, we have:
 \begin{itemize}
     \item \Achange{$|f(R)(\vec x; \vec y)| < m_f(\vec x, \vec y)$}
     \item \Achange{$f(R)(\vec x;\vec y) = f(\lambda |\vec u_r|,|\vec v_r| <m_f(\vec x,\vec y). r(\vec u_r;\vec v_r))_{r\in R} (\vec x;\vec y)$}
 \end{itemize}
\end{lem}
The first point is common to implicit complexity, being essentially Bellantoni and Cook's `polymax bounding lemma' from \cite{BellantoniCook}.
The second point expresses a dual property: while the first bounds the modulus of \emph{growth}, the second bounds the modulus of \emph{continuity}.

Here it is important the the new initial functions are relations, or at least that they have constant/limited growth rate. 
In fact, for the proof, cf.~\cite{CurziDas}, one needs a more complicated statement accounting for growth properties of the intermediate oracles $\vec a$ used for recursion, even though we only ultimately need the statement above for our purposes, once all such oracles $\vec a$ are `discharged'.

Using the Bounding Lemma
we have from \cite{CurziDas} (again accounting for further initial relations) the main characterisation result for $\bcpp$:
\begin{prop}
[Relativised characterisation]\label{prop:relativised-characterisation}
For a set $R$ of relations,
$\bcpp(\safnor R) \subseteq \fptime(R)$.
\end{prop}
The main point for proving this result is that the graph of $\permprefeq$ is relatively small, in particular for each $\vec y$ there are only polynomially many $\vec x \permprefeq \vec y$.\footnote{Of course, this polynomial depends on the length of $\vec y$, but for a given function of the algebra this is some global constant.}
So we can calculate a function $f(\vec x;\vec y) \in \bcpp(R)$ simply by polynomial-time induction (at the meta level) on $\permpref$, storing all previous values in a lookup table.
This table will have only polynomially many entries, by the previous observation about the size of the graph of $\permprefeq$, and each entry will have only polynomial size by the Bounding Lemma.

\subsection{$\cbc(\safnor \RR) \subseteq \bcpp(\RRsafnor)$: from circular proofs to recursive functions}
The point of $\bcpp$ in \cite{CurziDas} was to play the role of a target algebra to translate circular coderivations into.
The \emph{Translation Lemma} from that work~\cite[Lemma 47]{CurziDas},  accounting for further initial relations, gives:
\begin{lem}
[Relativised translation]\label{lem:rel-tran-lem}
Let $R$ be a set of relations. $\cbc (\safnor R) \subseteq \bcpp ( \safnor R)$.
\end{lem}
Note that we specialise the statement above only to sets of relations to avoid size issues potentially caused be new initial funtions of arbitrary growth rate. 
In fact, this proof requires closure of $\bcpp(F,\vec a)$ under a \emph{simultaneous} version of its recursion scheme, upon which a careful translation from circular coderivations in `cycle normal form' (see, e.g.,~\cite[Definition 6.2.1]{Brotherston-thesis}) to an equational specification can be duly resolved in $\bcpp$.

Now by setting $R = \RR$, we have the following consequence of \cref{prop:relativised-characterisation}:
\begin{cor}
\label{cor:nuB-in-fppoly}
$\nucbc \subseteq \FPpoly$
\end{cor}
\begin{proof}
We have $\nucbc\subseteq \cbc(\RRsafnor)$ by \cref{lem:non-wellfounded-oracles}, $\cbc(\RRsafnor)\subseteq \bcpp(\RRsafnor)$ by \cref{lem:rel-tran-lem}, $\bcpp(\RRsafnor)\subseteq \fptime(\RR)$ by \cref{prop:relativised-characterisation}, and finally $\fptime(\RR) \subseteq \FPpoly$ by \cref{prop:fppoly=fptime(RR)}.
\end{proof}

Along with \cref{cor:fppoly-in-nuB}, we have now established both directions of our main result \cref{thm:main-result} that $\nucbc = \FPpoly$, hence completing the proof.

	\section{Conclusions}

In this work we presented the two-sorted non-wellfounded proof system $\nucbc$ and proved that it characterises the complexity class $\FPpoly$. Our results build on previous work \cite{CurziDas}, where we defined the cyclic proof systems $\cbc$ and $\ncbc$ capturing, respectively,  $\fptime$ and $\felementary$~\cite{CurziDas}. 
The system $\nucbc$ is obtained from $\cbc$  by associating non-uniformity in computation to a form of non-wellfoundedness in proof theory. To establish the characterisation theorems, we also formalised some (presumably) folklore results on relativised function algebras for $\FPpoly$.

For future research, the first author is investigating non-wellfounded approaches to $\FPpoly$ in the setting of linear logic~\cite{Girard87}. In particular, we are studying a non-wellfounded  version of Mazza's Parsimonious Logic~\cite{Mazza15-log}, a variant of linear logic where the exponential modality $\oc$ satisfies Milner's law ($\oc A \simeq \oc A \otimes A)$.
This provides a natural computational interpretation of formulas $!A$ as types of streams on $A$. 
Mazza showed in~\cite{MazzaT15-non-uniform} that Parsimonious Logic can be used to capture $\Ppoly$ using \emph{wellfounded} proofs that are essentially  \emph{infinitely branching}. 
We conjecture that a similar characterisation can be obtained in a  non-wellfounded (and finitely branching) setting, using ideas from this work.

{Another \Achange{direction} is to explore applications of the results of this paper to probabilistic complexity. In particular, we aim to study fragments of $\nucbc$ modelling the class $\mathbf{BPP}$ (bounded-error probabilistic polynomial time), essentially by leveraging on well-known derandomisation methods showing the inclusion of $\mathbf{BPP}$ in $\FPpoly$, and hence in $\fptime(\RR)$ (see~\Cref{prop:fppoly=fptime(RR)}). A challenging aspect of this task is to obtain characterisation results that are entirely in the style of ICC, since  $\mathbf{BPP}$ is defined by explicit (error) bounds, as observed in~\cite{LagoT15}. We \Achange{suspect} that $\nucbc$ represents the right framework for investigating fully implicit characterisations of this class, where additional proof-theoretic conditions can be introduced to restrict computationally the access to oracles and, consequently, to model bounded-error probabilistic computation.}

As \cite{CurziDas} also established a system $\ncbc$ for $\felementary$, it would be pertinent to ask whether ideas in this work can be applied to $\ncbc$ to characterise $\FEpoly$, i.e., the class of functions computable in elementary time by a Turing machine with access to a polynomial advice. 
Unfortunately, the modulus of continuity established for $\ncbc$ in~\cite{CurziDas} is super-polynomial (indeed elementary), 
meaning that the same technique, a priori, would not restrict computation to only polynomial advice.
Consideration of this issue is left to future research.

\bibliography{main}

\begin{thebibliography}{10}

\bibitem{Arora-Barak}
Sanjeev Arora and Boaz Barak.
\newblock {\em Computational Complexity - {A} Modern Approach}.
\newblock Cambridge University Press, 2009.
\newblock URL:
  \url{http://www.cambridge.org/catalogue/catalogue.asp?isbn=9780521424264}.

\bibitem{baelde2016infinitary}
David Baelde, Amina Doumane, and Alexis Saurin.
\newblock Infinitary proof theory: the multiplicative additive case.
\newblock 62:42:1--42:17, 2016.
\newblock \href {https://doi.org/10.4230/LIPIcs.CSL.2016.42}
  {\path{doi:10.4230/LIPIcs.CSL.2016.42}}.

\bibitem{Bellantoni-fph}
Stephen Bellantoni.
\newblock Predicative recursion and the polytime hierarchy.
\newblock In Peter Clote and Jeffrey~B. Remmel, editors, {\em Feasible
  Mathematics II}, pages 15--29, Boston, MA, 1995. Birkh{\"a}user Boston.

\bibitem{BellantoniCook}
Stephen Bellantoni and Stephen Cook.
\newblock A new recursion-theoretic characterization of the polytime functions
  (extended abstract).
\newblock In {\em Proceedings of the Twenty-Fourth Annual ACM Symposium on
  Theory of Computing}, STOC '92, page 283–293, New York, NY, USA, 1992.
  Association for Computing Machinery.
\newblock \href {https://doi.org/10.1145/129712.129740}
  {\path{doi:10.1145/129712.129740}}.

\bibitem{BerardiT17}
Stefano Berardi and Makoto Tatsuta.
\newblock Equivalence of inductive definitions and cyclic proofs under
  arithmetic.
\newblock In {\em 32nd Annual {ACM/IEEE} Symposium on Logic in Computer
  Science, {LICS} 2017, Reykjavik, Iceland, June 20-23, 2017}, pages 1--12.
  {IEEE} Computer Society, 2017.
\newblock \href {https://doi.org/10.1109/LICS.2017.8005114}
  {\path{doi:10.1109/LICS.2017.8005114}}.

\bibitem{BerardiT19}
Stefano Berardi and Makoto Tatsuta.
\newblock Classical system of {M}artin-{L}of's inductive definitions is not
  equivalent to cyclic proofs.
\newblock {\em Log. Methods Comput. Sci.}, 15(3), 2019.
\newblock \href {https://doi.org/10.23638/LMCS-15(3:10)2019}
  {\path{doi:10.23638/LMCS-15(3:10)2019}}.

\bibitem{Brotherston05}
James Brotherston.
\newblock Cyclic proofs for first-order logic with inductive definitions.
\newblock In Bernhard Beckert, editor, {\em Automated Reasoning with Analytic
  Tableaux and Related Methods, International Conference, {TABLEAUX} 2005,
  Koblenz, Germany, September 14-17, 2005, Proceedings}, volume 3702 of {\em
  Lecture Notes in Computer Science}, pages 78--92. Springer, 2005.
\newblock \href {https://doi.org/10.1007/11554554\_8}
  {\path{doi:10.1007/11554554\_8}}.

\bibitem{Brotherston-thesis}
James Brotherston.
\newblock {\em Sequent calculus proof systems for inductive definitions}.
\newblock PhD thesis, 2006.
\newblock PhD thesis.

\bibitem{brotherston2011sequent}
James Brotherston and Alex Simpson.
\newblock Sequent calculi for induction and infinite descent.
\newblock {\em Journal of Logic and Computation}, 21(6):1177--1216, 2011.

\bibitem{buss98:intro-in-handbook}
Samuel~R. Buss.
\newblock Chapter i - an introduction to proof theory.
\newblock In Samuel~R. Buss, editor, {\em Handbook of Proof Theory}, volume 137
  of {\em Studies in Logic and the Foundations of Mathematics}, pages 1--78.
  Elsevier, 1998.
\newblock URL:
  \url{https://www.sciencedirect.com/science/article/pii/S0049237X98800165},
  \href {https://doi.org/https://doi.org/10.1016/S0049-237X(98)80016-5}
  {\path{doi:https://doi.org/10.1016/S0049-237X(98)80016-5}}.

\bibitem{CurziDas}
Gianluca Curzi and Anupam Das.
\newblock Cyclic implicit complexity.
\newblock {\em CoRR}, abs/2110.01114, 2021.
\newblock To appear in proceedings of \emph{LICS 2022}.
\newblock URL: \url{https://arxiv.org/abs/2110.01114}, \href
  {http://arxiv.org/abs/2110.01114} {\path{arXiv:2110.01114}}.

\bibitem{Das2021-preprint}
Anupam Das.
\newblock A circular version of {G}{\"{o}}del's {T} and its abstraction
  complexity.
\newblock {\em CoRR}, abs/2012.14421, 2020.
\newblock URL: \url{https://arxiv.org/abs/2012.14421}, \href
  {http://arxiv.org/abs/2012.14421} {\path{arXiv:2012.14421}}.

\bibitem{das2018logical}
Anupam Das.
\newblock On the logical complexity of cyclic arithmetic.
\newblock {\em Log. Methods Comput. Sci.}, 16(1), 2020.
\newblock \href {https://doi.org/10.23638/LMCS-16(1:1)2020}
  {\path{doi:10.23638/LMCS-16(1:1)2020}}.

\bibitem{Das2021}
Anupam Das.
\newblock On the logical strength of confluence and normalisation for cyclic
  proofs.
\newblock In Naoki Kobayashi, editor, {\em 6th International Conference on
  Formal Structures for Computation and Deduction, {FSCD} 2021, July 17-24,
  2021, Buenos Aires, Argentina (Virtual Conference)}, volume 195 of {\em
  LIPIcs}, pages 29:1--29:23. Schloss Dagstuhl - Leibniz-Zentrum f{\"{u}}r
  Informatik, 2021.
\newblock \href {https://doi.org/10.4230/LIPIcs.FSCD.2021.29}
  {\path{doi:10.4230/LIPIcs.FSCD.2021.29}}.

\bibitem{das2017cut}
Anupam Das and Damien Pous.
\newblock A cut-free cyclic proof system for {K}leene algebra.
\newblock In {\em International Conference on Automated Reasoning with Analytic
  Tableaux and Related Methods}, pages 261--277. Springer, 2017.

\bibitem{DP18}
Anupam Das and Damien Pous.
\newblock {Non-Wellfounded Proof Theory For
  (Kleene+Action)(Algebras+Lattices)}.
\newblock In Dan Ghica and Achim Jung, editors, {\em 27th EACSL Annual
  Conference on Computer Science Logic (CSL 2018)}, volume 119 of {\em Leibniz
  International Proceedings in Informatics (LIPIcs)}, pages 19:1--19:18,
  Dagstuhl, Germany, 2018. Schloss Dagstuhl--Leibniz-Zentrum fuer Informatik.
\newblock URL: \url{http://drops.dagstuhl.de/opus/volltexte/2018/9686}, \href
  {https://doi.org/10.4230/LIPIcs.CSL.2018.19}
  {\path{doi:10.4230/LIPIcs.CSL.2018.19}}.

\bibitem{dax2006proof}
Christian Dax, Martin Hofmann, and Martin Lange.
\newblock A proof system for the linear time $\mu$-calculus.
\newblock In {\em International Conference on Foundations of Software
  Technology and Theoretical Computer Science}, pages 273--284. Springer, 2006.

\bibitem{DeS19}
Abhishek De and Alexis Saurin.
\newblock Infinets: The parallel syntax for non-wellfounded proof-theory.
\newblock In Serenella Cerrito and Andrei Popescu, editors, {\em Automated
  Reasoning with Analytic Tableaux and Related Methods - 28th International
  Conference, {TABLEAUX} 2019, London, UK, September 3-5, 2019, Proceedings},
  volume 11714 of {\em Lecture Notes in Computer Science}, pages 297--316.
  Springer, 2019.
\newblock \href {https://doi.org/10.1007/978-3-030-29026-9\_17}
  {\path{doi:10.1007/978-3-030-29026-9\_17}}.

\bibitem{fortier2013cuts}
J{\'e}r{\^o}me Fortier and Luigi Santocanale.
\newblock Cuts for circular proofs: semantics and cut-elimination.
\newblock In {\em Computer Science Logic 2013 (CSL 2013)}. Schloss
  Dagstuhl-Leibniz-Zentrum fuer Informatik, 2013.

\bibitem{Girard87}
Jean{-}Yves Girard.
\newblock Linear logic.
\newblock {\em Theor. Comput. Sci.}, 50:1--102, 1987.
\newblock \href {https://doi.org/10.1016/0304-3975(87)90045-4}
  {\path{doi:10.1016/0304-3975(87)90045-4}}.

\bibitem{HainryMP20}
Emmanuel Hainry, Damiano Mazza, and Romain P{\'{e}}choux.
\newblock Polynomial time over the reals with parsimony.
\newblock In Keisuke Nakano and Konstantinos Sagonas, editors, {\em Functional
  and Logic Programming - 15th International Symposium, {FLOPS} 2020, Akita,
  Japan, September 14-16, 2020, Proceedings}, volume 12073 of {\em Lecture
  Notes in Computer Science}, pages 50--65. Springer, 2020.
\newblock \href {https://doi.org/10.1007/978-3-030-59025-3\_4}
  {\path{doi:10.1007/978-3-030-59025-3\_4}}.

\bibitem{Hofmann97}
Martin Hofmann.
\newblock A mixed modal/linear lambda calculus with applications to
  {B}ellantoni-{C}ook safe recursion.
\newblock In Mogens Nielsen and Wolfgang Thomas, editors, {\em Computer Science
  Logic, 11th International Workshop, {CSL} '97, Annual Conference of the
  EACSL, Aarhus, Denmark, August 23-29, 1997, Selected Papers}, volume 1414 of
  {\em Lecture Notes in Computer Science}, pages 275--294. Springer, 1997.
\newblock \href {https://doi.org/10.1007/BFb0028020}
  {\path{doi:10.1007/BFb0028020}}.

\bibitem{Kleene71:intro-to-metamath}
Stephen~Cole Kleene.
\newblock {\em Introduction to Metamathematics}.
\newblock Bubliotheca Mathematica. Wolters-Noordhoff Publishing, 7 edition,
  1971.

\bibitem{Ko}
Ker-I Ko.
\newblock {\em Complexity Theory of Real Functions}.
\newblock Birkhauser Boston Inc., USA, 1991.

\bibitem{Kuperberg-Pous21}
Denis Kuperberg, Laureline Pinault, and Damien Pous.
\newblock Cyclic proofs, system {T}, and the power of contraction.
\newblock {\em Proc. {ACM} Program. Lang.}, 5({POPL}):1--28, 2021.
\newblock \href {https://doi.org/10.1145/3434282} {\path{doi:10.1145/3434282}}.

\bibitem{LagoT15}
Ugo~Dal Lago and Paolo~Parisen Toldin.
\newblock A higher-order characterization of probabilistic polynomial time.
\newblock {\em Inf. Comput.}, 241:114--141, 2015.
\newblock \href {https://doi.org/10.1016/j.ic.2014.10.009}
  {\path{doi:10.1016/j.ic.2014.10.009}}.

\bibitem{Leivant91}
Daniel Leivant.
\newblock A foundational delineation of computational feasiblity.
\newblock In {\em Proceedings of the Sixth Annual Symposium on Logic in
  Computer Science {(LICS} '91), Amsterdam, The Netherlands, July 15-18, 1991},
  pages 2--11. {IEEE} Computer Society, 1991.
\newblock \href {https://doi.org/10.1109/LICS.1991.151625}
  {\path{doi:10.1109/LICS.1991.151625}}.

\bibitem{Mazza14}
Damiano Mazza.
\newblock Non-uniform polytime computation in the infinitary affine
  lambda-calculus.
\newblock In Javier Esparza, Pierre Fraigniaud, Thore Husfeldt, and Elias
  Koutsoupias, editors, {\em Automata, Languages, and Programming - 41st
  International Colloquium, {ICALP} 2014, Copenhagen, Denmark, July 8-11, 2014,
  Proceedings, Part {II}}, volume 8573 of {\em Lecture Notes in Computer
  Science}, pages 305--317. Springer, 2014.
\newblock \href {https://doi.org/10.1007/978-3-662-43951-7\_26}
  {\path{doi:10.1007/978-3-662-43951-7\_26}}.

\bibitem{Mazza15-log}
Damiano Mazza.
\newblock Simple parsimonious types and logarithmic space.
\newblock In Stephan Kreutzer, editor, {\em 24th {EACSL} Annual Conference on
  Computer Science Logic, {CSL} 2015, September 7-10, 2015, Berlin, Germany},
  volume~41 of {\em LIPIcs}, pages 24--40. Schloss Dagstuhl - Leibniz-Zentrum
  f{\"{u}}r Informatik, 2015.
\newblock \href {https://doi.org/10.4230/LIPIcs.CSL.2015.24}
  {\path{doi:10.4230/LIPIcs.CSL.2015.24}}.

\bibitem{MazzaT15-non-uniform}
Damiano Mazza and Kazushige Terui.
\newblock Parsimonious types and non-uniform computation.
\newblock In Magn{\'{u}}s~M. Halld{\'{o}}rsson, Kazuo Iwama, Naoki Kobayashi,
  and Bettina Speckmann, editors, {\em Automata, Languages, and Programming -
  42nd International Colloquium, {ICALP} 2015, Kyoto, Japan, July 6-10, 2015,
  Proceedings, Part {II}}, volume 9135 of {\em Lecture Notes in Computer
  Science}, pages 350--361. Springer, 2015.
\newblock \href {https://doi.org/10.1007/978-3-662-47666-6\_28}
  {\path{doi:10.1007/978-3-662-47666-6\_28}}.

\bibitem{mints1978finite}
Grigori~E Mints.
\newblock Finite investigations of transfinite derivations.
\newblock {\em Journal of Soviet Mathematics}, 10(4):548--596, 1978.

\bibitem{niwinski1996games}
Damian Niwi{\'n}ski and Igor Walukiewicz.
\newblock Games for the $\mu$-calculus.
\newblock {\em Theoretical Computer Science}, 163(1-2):99--116, 1996.

\bibitem{Simpson17}
Alex Simpson.
\newblock Cyclic arithmetic is equivalent to peano arithmetic.
\newblock In Javier Esparza and Andrzej~S. Murawski, editors, {\em Foundations
  of Software Science and Computation Structures - 20th International
  Conference, {FOSSACS} 2017, Held as Part of the European Joint Conferences on
  Theory and Practice of Software, {ETAPS} 2017, Uppsala, Sweden, April 22-29,
  2017, Proceedings}, volume 10203 of {\em Lecture Notes in Computer Science},
  pages 283--300, 2017.
\newblock \href {https://doi.org/10.1007/978-3-662-54458-7\_17}
  {\path{doi:10.1007/978-3-662-54458-7\_17}}.

\end{thebibliography}

\clearpage
\appendix

\section{Some properties of progressing  $\bcnorec$-coderivations established in~\cite{CurziDas}}

\noindent
\textbf{Proof of~\Cref{prop:prog-imp-tot}.}
 We proceed by contradiction.
If $\denot \der$ is non-total then, since each rule preserves totality top-down, we must have that $\denot{\der'}$ is non-total for one of $\der$'s immediate sub-coderivations $\der'$.
Continuing this reasoning we can build an infinite leftmost `non-total' branch $B=(\der^{i})_{i<\omega}$.
Let $(\sn^i)_{i\geq k}$ be a progressing thread along $B$, and assign to each $\sn^i$ the least natural number $n_i \in \Nat$ such that $\denot{\der^{i}}$ is non-total when $n_i$ is assigned to the type occurrence $\sn^i$.

Now, notice that:
\begin{itemize}
    \item $(n_i)_{i\geq k}$ is monotone non-increasing, by inspection of the rules and their interpretations from \Cref{defn:semantics-bc}.
    \item $(n_i)_{i\geq k}$ does not converge, since $(\sn^i)_{i\geq k}$ is progressing and so is infinitely often principal for $\cnd_\sq$ or $\cndl_\sq$, where the value of $n_i$ must strictly decrease (cf., again, \Cref{defn:semantics-bc}).
\end{itemize}
This contradicts the well-ordering property of the natural numbers.
\qed

\begin{prop}[\cite{CurziDas}]\label{prop:cutbox} Given a progressing  $\bcnorec$-coderivation  $\der: \sq \Gamma , \vec \n \seqar \sn$, there is a progressing  $\bcnorec$-coderivation $\der^*:\sq\Gamma \seqar \sn$   such that:
\begin{equation*}
    \model{\der}(\vec{x}; \vec{y})= \model{\der^*}(\vec{x};).
\end{equation*}
\end{prop}
\begin{proof} By progressiveness, any infinite branch contains a $\cnd_\sq$-step or a $\cndl_\sq$-step, which have non-modal succedents. 
Thus there is a set of $\cnd_{\sq}$-occurrences and $\cndl_\sq$-occurrences that forms a bar across $\der$.
By K\"{o}nig Lemma, the set of all nodes of $\der$ below this bar, say $X_\der$, is finite. The proof now follows by induction on the cardinality of $X_\der$. The case where the last rule of $\der$ is an instance of $\id, \zero, 1, \sqr, \cnd_{\n}, \cnd_\sq, \cndl_{\n}, \cndl_\sq$, or $\srec$ are trivial. If  the last rule of $\der$ is an instance of $\exch_{\n},\exch_{\sq},  \sql,\succ i$, and $\wk_\sq$ then we apply the induction hypothesis. Let us now suppose that $\der$ has been obtained from a derivation $\der_0$ by applying an instance of  $\wk_\n$. By induction hypothesis, there exists a derivation  $\pder0^*:\lists{n}{\sn}{\sn} \seqar \sn$ such that $\model{\pder0}(\vec x; \vec y)= \model{\pder0^*}(\vec x; )$. Since $\model{\der}(\vec x; \vec y, y)= \model{\pder0}(\vec x; \vec y)=\model{\pder0^*}(\vec x; )$ we just set $\der^*= \pder{0}^*$.  Suppose now that $\der$ is obtained from two derivations $\der_0$ and $\der_1$ by applying an instance of  $\cut_\n$. By induction hypothesis, there exists $\pder{1}^*$ such that $\model{\pder{1}}(\vec x;\vec y,y)=\model{\pder{1}^*}(\vec x;)$. Since 
 $\model{\pder{1}}(\vec x; \vec y,  \model{\pder{0}}(\vec x; \vec y))= \model{\pder{1}^*}(\vec x;)$, we set $\der^*=\pder{1}^*$. As for the case where the last rule is $\cut_\sq$, by induction hypothesis, there exist derivations $\pder{0}^*$ and $\pder{1}^*$ such that $\model{\pder{0}}(\vec x; \vec y)= \model{\pder{0}^*}(\vec x;)$ and $\model{\pder{1}}(\vec x,x; \vec y)= \model{\pder{1}^*}(\vec x, x;)$, so that we define $\der^*$ as the derivation obtained from $\pder{0}^*$ and $\pder{1}^*$ by applying the rule $\cut_{\sq}$.
\end{proof}

\section{Simultaneous recursion schemes}\label{sec:simult-rec}

For our main results, we will ultimately need that the function algebra \black{$\bcpp(F)$} is closed under simultaneous versions of its recursion scheme.

\begin{defn}[Simultaneous schemes]
\label{defn:simultaneous-schemes}

\black{Let $F$ be a set of two-sorted functions.} We define the scheme $\ssrecpp$ as follows, for arbitrary $\vec a = a_1, \dots, a_{k}$:
\begin{itemize}
    \item from $h_i(\vec a)(\vec x; \vec y)$ over $\vec a, \vec b$\black{, $F$}, for $1 \leq i\leq k$, define $f_i(\vec x;\vec y)$ over $\vec b$\black{, $F$}, for $1 \leq i \leq k$, by:
    $$
    f_i(\vec x; \vec y) = 
    h_i((\lambda \vec u \permpref \vec x, \lambda \vec v\subseteq \vec y. f_j(\vec u; \vec v))_{1 \leq j\leq  k})(\vec x; \vec y)
    $$
\end{itemize}
\end{defn}
\begin{prop} \label{prop:simultaneous-recursion-admissible}
 If $\vec f(\vec x;\vec y)$ over $\vec b$\black{,$F$} are obtained by applying  $\ssrecpp$ to $\vec h(\vec a)(\vec x;\vec y) \in \bcpp(\vec a,\vec b\black{, F})$, then also $\vec f(\vec x;\vec y) \in \bcpp (\vec b\black{,  F})$
\end{prop}
\begin{proof}
\black{In what follows we shall neglect relativisation to $F$, as these oracles play no role in the proof.} Let $f_i(\vec x;\vec y)$ and $h_i(a_1, \dots, a_k)(\vec x; \vec y)$ be as given in \Cref{defn:simultaneous-schemes}, and temporarily write $f_j^{\vec x;\vec y}$ for $\lambda \vec u \permpref \vec x, \lambda \vec v \permprefeq \vec y . f_j (\vec u;\vec v)$, so we have:
$$
f_i(\vec x;\vec y) = h_i (f_1^{\vec x; \vec y}, \dots, f_k^{\vec x;\vec y})(\vec x; \vec y)
$$
For $i\in \Nat$, let us temporarily write $\numeral i$ for $i$ in binary notation\footnote{In fact, any notation will do, but we pick one for concreteness.}, and $\permi i$ for the list $$\numeral i, \numeral{i+1}, \dots, \numeral k, \numeral 1, \numeral 2, \dots, \numeral{i-1}$$
Note that, for all $i=1, \dots, k$, $\permi i$ is a permutation (in fact a rotation) of $\numeral 1, \dots, \numeral k$.

Now, let $f(\vec x; \vec y, \vec z)$ over oracles $\vec b$ be given as follows:
\begin{equation}
    \label{eq:sim-rec-reduction}
    f(\vec x; \vec y, \vec z) \dfn \begin{cases}
h_1 (f_1^{\vec x;\vec y}, \dots, f_k^{\vec x; \vec y})(\vec x;\vec y) & \vec z = \permi 1 \\
\vdots \\
h_k (f_1^{\vec x;\vec y}, \dots, f_k^{\vec x; \vec y})(\vec x;\vec y) & \vec z = \permi k \\
0 & \text{otherwise}
\end{cases}
\end{equation}
Note that this really is a finite case distinction since each of the boundedly many $\permi i$ has bounded size, both bounds depending only on $k$, and so is computable in $\bcnorec$ over $\vec h$. 

By definition, then, we have that $f(\vec x;\vec y, \permi i) = f_i (\vec x; \vec y)$.
Moreover note that, for each $j=1,\dots,k$, we have,
$$
\arraycolsep=2pt
\begin{array}{rcl}
   f_j^{\vec x;\vec y} (\vec u'; \vec v')  &=& (\lambda \vec u \permpref \vec x, \lambda \vec v \permprefeq \vec y. f_j(\vec u;\vec v)) (\vec u';\vec v') \\
    & = & (\lambda \vec u \permpref \vec x, \lambda \vec v \permprefeq \vec y. f(\vec u;\vec v,\permi j ))(\vec u';\vec v') \\
    & = & (\lambda \vec u \permpref \vec x, \lambda \vec v \permprefeq \vec y, \lambda \vec w \permprefeq \vec z . f(\vec u;\vec v, \vec w)) (\vec u';\vec v',\permi j)
\end{array}
$$
as long as $\vec z$ is some $\vec i$,
so indeed~\eqref{eq:sim-rec-reduction} has the form,
$$
f(\vec x; \vec y, \vec z) = h (\lambda \vec u \permpref \vec x, \lambda \vec v \permprefeq \vec y, \lambda \vec w \permprefeq \vec z . f(\vec u; \vec v, \vec w)) (\vec x;\vec y)
$$
and $f(\vec x;\vec y,\vec z) \in \bcpp(\vec b)$ by $\srecpp$.
Finally, since  $f_i(\vec x;\vec y) = f(\vec x;\vec y, \permi i)$, we indeed have that each $f_i(\vec x;\vec y) \in \bcpp(\vec b)$.
\end{proof}
\section{Proof of~\Cref{lem:rel-tran-lem}}
\black{In this section we prove~\Cref{lem:rel-tran-lem}.  In what follows, when clear from the context, we shall simply write $\der$ in place of $\der(F)$ to facilitate readability. We recall that oracles in a coderivation are initial sequents.}

First, we observe that a regular coderivation (with oracles) can be naturally seen as a finite tree with `backpointers', a representation known as \emph{cycle normal form}, cf.~\cite{Brotherston05,brotherston2011sequent}.   

\begin{defn}[Cycle normal form]
\black{Let $F$ be a set of  two-sorted functions, and let $\der$ be a regular $\bcnorec(F)$-coderivation. }
The  \emph{cycle normal form} (or simply \emph{cycle nf}) of $\der$ is a pair $\derrd$, where  $\rd$ is a partial self-mapping on the nodes of $\der$ whose domain of definition is denoted $\bud(\der)$ and:
\begin{enumerate}[(i)]
    \item every infinite branch of $\der$ contains some (unique) $\nu \in \bud(\der)$;
    \item if $\nu \in \bud(\der)$ then both $\rd(\nu) \sqsubset \nu$ and  $\der_{\rd(\nu)} = \der_\nu$;
    \item for any two distinct nodes $\mu \sqsubset \nu$ strictly below $\bud(\der)$, $\der_\mu \neq \der_\nu$
\end{enumerate}

We call any $\nu \in \bud(\der)$ a \emph{bud}, and $\rd(\nu)$ its \emph{companion}. A \emph{terminal} node is either one of the leaves of $\der$ (\black{among which, the functions in  $F$}) or   a bud. The set of nodes of $\der$ bounded above by a terminal node is denoted $T_\der$.  
Given a node $\nu\in T_\der$, we  define $\bud_\nu(\der)$ as the restriction of buds to those above $\nu$.
\end{defn}

\begin{rem}\label{rem:canonicalnodesfinite}
The cycle normal form of a regular coderivation $\der$ always exists, as  by definition any infinite branch contains a node $\nu$ such that $\pder \nu=\pder \mu$ for some node $\mu$ below $\nu$. 
$\bud(\der)$ is designed to consist of just the \emph{least} such nodes, so that by construction the  cycle normal form is unique. 
Note that $\bud(\der)$ must form an antichain: if $\mu, \nu \in \bud(\der)$ with $\mu\sqsubset \nu$, then $\rd(\mu)\sqsubset \mu $ are below $\bud(\der)$ but we have $\der_{\rd(\mu)} = \der_\mu$ by (ii) above, contradicting the (iii). 

Also, notice that any branch of $\der$ contains a leaf of $T_\der$. Moreover, since $\bud(\der)$ is an antichain, the leaves of $T_\der$ defines a ‘bar’ across  $\der$,
and so $T_\der$ is a finite tree. 
\end{rem}


The following proposition allows us to reformulate   progressiveness, safety and left-leaning conditions for cycle normal forms. 
\begin{prop}\label{prop:structureofcycles}
\black{Let $F$ be  two-sorted  functions, and let $\der$ be a regular   $\bcnorec(F)$-coderivation} with cycle nf $\derrd$. 
For any  $\nu \in \bud(\der)$, the (finite) path $\pi$ from $\rd(\nu)$ to $\nu$ satisfies:
\begin{enumerate}
\item \label{enum:budcompanion0} if  $\der$ is progressing, $\pi$ must contain the conclusion of an instance of $\cnd_{\sn}$ \black{or $\cndl_\sq$};
    \item \label{enum:budcompanion1} if  $\der$ is \black{progressing and safe}, $\pi$ cannot  contain  the conclusion of  $\cut_{\sn}$, $\sql$,  $\wk_{\sq}$, and  the   leftmost premise of  $\cnd_{\sq}$ \black{and of $\cndl_{\sq}$};
    \item \label{enum:budcompanion3} if $\der$ is a \black{$\cbc(F)$-coderivation},  $\pi$  cannot contain the conclusion of $\wk_{\n}$, the leftmost premise of  $\cnd_{\n}$ \black{and of $\cndl_{\n}$}, and the rightmost premise of $\cut_{\n}$.
\end{enumerate}
\end{prop}
\begin{proof}
By definition of cycle nf, each path from $\rd(\nu)$ to $\nu$ in $\derrd$ is contained in a branch of $\der$ such that each rule instance in the former appears  infinitely many times in the latter. Hence:
\begin{enumerate}
\item[(i)]  if $\der$ is progressing,  the path contains the conclusion of an instance of $\cnd_{\sn}$   \black{or $\cndl_\sq$};
\item[(ii)]  if $\der$ is safe, the path   cannot contain the conclusion of a $\cut_{\sn}$ rule;
\item[(iii)] if $\der$ is left-leaning, the path    cannot contain the rightmost premise of a $\cut_{\n}$ rule.
\end{enumerate}
This shows point~\ref{enum:budcompanion0}. Let us consider  point~\ref{enum:budcompanion1}.  By point~(ii), if $\der$ is safe then,  going from a node $\mu$ of the path to each of its children $\mu'$,  the number of modal formulas in the context of the corresponding sequents cannot increase.  Moreover, the only cases where this number strictly decreases is when $\mu$ is the conclusion of $\sql$,  $\wk_{\sn}$, or  when $\mu'$ is  the leftmost premise of $\cnd_{\sq}$ \black{and of $\cndl_{\sq}$}. Since $\rd(\nu)$ and $\nu$ must be labelled with the same sequent, all such cases are impossible. As for  point~\ref{enum:budcompanion3} we notice that,   by point~(iii) and the above reasoning, if $\der$ is safe and left-leaning then,  going from a node $\mu$ of the path to each of its children $\mu'$,  the number of non-modal formulas in the context of the corresponding sequents cannot increase.  Moreover, the only cases where this number strictly decreases is when $\mu$ is the conclusion of $\wk_{\n}$, or  when $\mu'$ is  the leftmost premise of $\cnd_{\n}$ \black{and of $\cndl_{\n}$}. Since $\rd(\nu)$ and $\nu$ must be labelled with the same sequent, all such cases are impossible. 
\end{proof} 

In what follows we shall indicate circularities in cycle nfs explicitly by extending   \black{$\cbc(F)$} with a new inference rule called $\com$: 
$$
\vlinf{\com}{X}{\Gamma \Rightarrow A}{\Gamma \Rightarrow A}
$$
where $X$ is a finite set of nodes. 
In this presentation, we insist that each companion $\nu$ of the cycle nf $\derrd$ is always the conclusion of an instance of $\com$, where $X$ denotes the set of buds $\nu'$ such that $\rd(\nu')=\nu$.  
This expedient will allow us to formally 
distinguish cases when a node of $T_\der$ is a companion from those where it is the conclusion of a standard rule of $\bcnorec$.

To facilitate the translation, we shall define two disjoint  sets  $\close{\nu}$ and $\open{\nu}$. Intuitively, given a cycle nf $\derrd$ 
 and  $\nu\in T_\der$,  $\close{\nu}$ is the set of  companions above  $\nu$, while   $\open{\nu}$  is the set of  buds whose companion is strictly below   $\nu$.

\begin{defn} \black{Let $F$ be a set of   two-sorted  functions, and let}  $\derrd$ be the  cycle normal form of a \black{$\bcnorec(F)$-coderivation} $\der$. We define the following two sets for any $\nu\in T_\der$:
$$
\begin{aligned}
\close{\nu}&\dfn \{ \mu \in \rd(\bud_\nu(\der))  \ \vert \ \nu \sqsubseteq \mu \}\\
\open{\nu}&\dfn \{ \mu \in \bud_\nu(\der)) \ \vert \ \rd(\mu) \sqsubset \nu \}
\end{aligned}
$$
\end{defn}

\black{We are now ready to prove~\Cref{lem:rel-tran-lem}, i.e.,  $\cbc (\safnor R) \subseteq \bcpp ( \safnor R)$ for $R$  a set of relations. Actually, we shall establish a stronger result,~\Cref{lem:translation} below. Given $\der \in \cbc (\safnor R)$,} the  proof proceeds by analysing each node $\hnu\in T_\der$  and associates with it an instance of the scheme  $\ssrecpp$ (\Cref{defn:simultaneous-schemes})  that simultaneously defines the  functions $\{ \model{\pder\nu} \ \vert \ \nu\in \close \nu \cup \{\hnu \}\}$,   with the help of an additional set of oracles $\{ \model{\pder\mu} \ \vert \ \mu\in \open{\hnu}\}$. When $\hnu$ is the root of $\der$, note that $\open{\hnu} = \emptyset$ and so the function thus defined will be oracle-free.
 Thus we obtain an instance of   $\ssrecpp$ defining $\model{\der}$, and so \black{$\model{\der} \in \bcpp(\safnor R)$} by Proposition~\ref{prop:simultaneous-recursion-admissible}.

\begin{lem}[Translation Lemma, general version]\label{lem:translation} \black{Let $R$ be a set of relations. } Given $\derrd$ the cycle nf of a \black{$\cbc(\safnor R)$-coderivation} $\der$, and  $\hnu\in T_\der$:
\begin{enumerate}
    \item \label{enum:1} \black{If  $\open \hnu
   =\emptyset$ then   $\model{\pder \hnu} \in  \bcpp(\safnor R)$. }
    \item \label{enum:2} If  $\open \hnu \neq \emptyset$  then $\forall \nu \in \close  \hnu \cup \{ \hnu\}$:
 $$
    \begin{aligned}
     \model{\pder \nu}(\vec{x}; \vec{y})&=h_{\nu}(( \lambda \vec{u}\subseteq \vec{x},\black{\lambda \vec v \subseteq \vec y}. \model{\pder \mu}(\vec{u};\vec v))_{\mu \in \close \hnu \cup \open \hnu})(\vec{x}; \vec{y})
    \end{aligned}
$$
    where:
    \begin{enumerate}
    \item \label{enum:a} $h_\nu \in \bcpp( (\model{\pder{\mu}})_{\mu \in \open \hnu}, (a_\mu)_{\mu \in \close \hnu}\black{, \safnor R})$ and hence    $\model{\pder \nu}\in  \bcpp( (\model{\pder \mu})_{\mu \in \open \hnu}\black{, \safnor R})$;
    \item \label{enum:b}  the order $\vec u \subseteq \vec x$ is strict  if either $\nu, \mu \in \close \hnu$ or the path from $\nu$ to $\mu$ in $\pder \hnu$ contains the conclusion of an instance of  $\cnd_{\sq}$ \black{or an instance of $\cndl_\sq$};
    \end{enumerate}
    \end{enumerate}
\end{lem}

\begin{proof}
Points~\ref{enum:1} and~\ref{enum:2} are proven by simultaneous induction on the longest distance of $\nu_0$ from a leaf of  $T_\der$.  Notice that in the following situations only point~\ref{enum:1} applies:
\begin{itemize}
\item  \black{$\nu$ is the conclusion of an initial sequent $r \in R$.}
    \item  $\nu$ is the conclusion of an instance of $\id$, $\zero$, \black{or $1$};
    \item  $\nu$ is the conclusion of an instance of $\wk_{\n}$, $\wk_{\sq}$, $\sql$ and  $\cut_{\sn}$;
\end{itemize}
In particular, the  last two cases hold by  Proposition~\ref{prop:structureofcycles}.\ref{enum:budcompanion1}-\ref{enum:budcompanion3}, as it must be that $\open{\nu'}=\emptyset$ for any premise $\nu'$ of $\hnu$. Let us discuss the case where  $\hnu$ is the conclusion of a  $\cut_{\sq}$ rule with premises $\nu_1$ and $\nu_2$.  By induction on point~\ref{enum:1} we have $\model{\pder {\nu_1}}(\vec x; \vec y), \model{\pder{\nu_2}}(\vec x, x; \vec y)\in \black{\bcpp(\safnor R)}$. Since the conclusion of $\pder{\nu_1}$  has modal succedent, by  Proposition~\ref{prop:cutbox}  there must be a coderivation $\der^*$ such that $\model{\der^*}(\vec x;)= \model{\pder{\nu_1}}(\vec x; \vec y)\in \black{\bcpp(\safnor R)}$.  Hence, we define $\model{\pder \hnu}(\vec x; \vec y)= \model{\pder{\nu_2}}(\vec x, \model{\der^*}(\vec x;); \vec y)\in \black{\bcpp(\safnor R)}$. 

Let us now  consider   point~\ref{enum:2}. If  $\hnu$ is the conclusion of a bud  then  $\open \hnu= \{ \hnu\}$, $\close \hnu=\emptyset$, and all points hold trivially. The cases where $\hnu$ is an instance of {an initial sequent \black{$r\in R$},} $\wk_\n$, $\exch_\n$, $\exch_\sq$, $\sqr$,  $\succ 0$ or $\succ 1$ are straightforward.  Suppose that $\hnu$ is the conclusion of a $\cnd_\sq$ step with premises $\nu'$, $\nu_1$, and $\nu_2$, and let us  assume $\open{\nu_1} \neq \emptyset$, $\open{\nu_2} \neq \emptyset$. By Proposition~\ref{prop:structureofcycles}.\ref{enum:budcompanion1} we have $\open{\nu'}=\emptyset$, so that \black{$\model{\pder{\nu'}}\in \bcpp(\safnor R)$} by induction hypothesis on point~\ref{enum:1}.  By definition,  $\open \hnu=\open{\nu_1}\cup \open{\nu_2}$ and  $\close \hnu=\close{\nu_1}\cup \close{\nu_2}$. Then, we set:
\[
\arraycolsep=2pt
    \begin{array}{c}
    \model{\pder \hnu}(x, \vec{x}; \vec{y})
= \cnd(;x,  \model{\pder{\nu'}}(\vec{x}; \vec{y}), \model{\pder{\nu_1}}(\pred(x;),\vec{x}; \vec{y}) , \model{\pder{\nu_2}}(\pred(x;), \vec{x}; \vec{y}))
  \end{array}
\]
  where $\pred(x;)$ can be defined from $\pred(;x)$ and projections.  By induction hypothesis on $\nu_i$:
  \[
  \arraycolsep=2pt
  \def\arraystretch{1.2}
      \begin{array}{rcl}
       \model{\pder{\nu_i} }(\pred(x;), \vec{x}; \vec{y}) &
       = & 
       h_{\nu_i}(( \lambda u,\vec{u}\subseteq \pred(x;),\vec{x}, \black{\lambda \vec{v} \subseteq \vec y}. \model{\pder{\mu}}(u,\vec{u};\vec{v}))_{\mu \in \close{\nu_i}\cup \open{\nu_i}})(\pred(x;),\vec{x}; \vec{y})\\
       & =& 
       h_{\nu_i}(( \lambda u,\vec{u}\subset x,\vec{x}, \black{\lambda \vec{v} \subseteq \vec y}. \model{\pder{\mu}}(u,\vec{u};\vec{v}))_{\mu \in \close{\nu_i}\cup \open{\nu_i}})(\pred(x;),\vec{x}; \vec{y})
      \end{array}
\]
  hence 
  \[
  \arraycolsep=2pt
  \begin{array}{c}
        \model{\pder \hnu}(x, \vec{x}; \vec{y})
        =
        h_{\hnu}(( \lambda u,\vec{u}\subset x,\vec{x}, \black{\lambda \vec v \subseteq\vec y}. \model{\pder{\mu}}(u,\vec{u};\vec v))_{\mu \in \close{\hnu}\cup \open{\hnu}})(x,\vec{x}; \vec{y}) 
  \end{array}
  \]
  for some $h_{\hnu}$. By IH,   $h_{\hnu}\in \bcpp((\model{\pder \mu})_{\mu \in \open \hnu}, (a_\mu)_{\mu \in \close \hnu}\black{, \safnor R})$  and $\model{\pder \hnu}\in \bcpp((\model{\pder \mu})_{\mu \in \open \hnu}\black{, \safnor R})$. This shows point~\ref{enum:a}. Point~\ref{enum:b} is trivial.
 
Let us now consider the case where $\hnu$ is an instance of $\cnd_\n$, assuming $\open{\nu_1} \neq \emptyset$ and $\open{\nu_2} \neq \emptyset$. By Proposition~\ref{prop:structureofcycles}.\ref{enum:budcompanion3} we have $\open{\nu'}=\emptyset$, so that $\model{\pder{\nu'}}\in \black{\bcpp(\safnor R)}$ by induction hypothesis on point~\ref{enum:1}.  By definition,  $\open \hnu=\open{\nu_1}\cup \open{\nu_2}$ and  $\close \hnu=\close{\nu_1}\cup \close{\nu_2}$. Then, we set:
\[
\arraycolsep=2pt
\begin{array}{c}
 \model{\pder \hnu}(\vec{x};y,  \vec{y})
 = 
 \cnd(;y, \model{\pder{\nu'}}(\vec{x}; \vec{y}), \model{\pder{\nu_1}}(\vec{x};\pred(;y), \vec{y}) , \model{\pder{\nu_2}}( \vec{x};\pred(;y), \vec{y}))    
\end{array}
\]
By induction hypothesis on $\nu_i$:
  \[
\arraycolsep=2pt
\begin{array}{rcl}
\model{\pder{\nu_i} }(\vec{x}; \pred(;y), \vec{y}) 
&=&
h_{\nu_i}(( \lambda \vec{u}\subseteq x,\vec{x}, \lambda v,\vec{v}\subseteq \vec z. \model{\pder{\mu}}(\vec{u};v,\vec{v}))_{\mu \in \close{\nu_i}\cup \open{\nu_i}})(\vec{x};\vec z)\\
&=& 
h_{\nu_i}(( \lambda \vec{u}\subset \vec{x}, \lambda v,\vec{v}\subset y,\vec{y}. \model{\pder{\mu}}(\vec{u};v,\vec{v}))_{\mu \in \close{\nu_i}\cup \open{\nu_i}})(\vec{x};\vec z)
\end{array}
\]
where $\vec z= \pred(;y), \vec{y}$, and   hence 
   \[
\arraycolsep=2pt
\begin{array}{c}
\model{\pder \hnu}( \vec{x};y, \vec{y})
=
h_{\hnu}((\lambda \vec{u}\subseteq \vec{x}, \lambda v,\vec v\subset y, \vec y. \model{\pder{\mu}}(\vec{u};v,\vec v))_{\mu \in \close{\hnu}\cup \open{\hnu}})(\vec{x};y, \vec{y})
\end{array}
\]
 for some $h_{\hnu}$. Point~\ref{enum:a} and~\ref{enum:b} are given by the induction hypothesis.

\black{The cases for $\cndl_\sq$ and $\cndl_\n$ are similar.}
  
Let us now consider the case where $\hnu$ is the conclusion of an instance of  $\com$  with premise $\nu'$, where  $X$ is the set of nodes labelling the rule. We have  $\open{\hnu}= \open{\nu'} \setminus X$ and $\close{\hnu}=\close{\nu'}\cup\{\hnu \}$. We want to find $(h_\nu)_{\nu \in \close \nu \cup \{ \hnu\}}$ defining the equations for $(\model{\pder \nu})_{\nu \in \close \nu \cup \{ \hnu\}}$   in such a way that points~\ref{enum:a}-\ref{enum:b} hold. We shall start by defining $h_{\hnu}$. First, note that, by definition of cycle nf, $\model{\pder{\hnu}}(\vec{x}; \vec{y}) =\model{\pder{\nu'}}(\vec{x}; \vec{y})=\model{\pder \mu}(\vec{x}; \vec{y})$ for all $\mu \in X$.  By  induction hypothesis on $\nu'$ there exists a family  $(g_{\nu})_{\nu \in \close{\nu'}\cup \{ \nu'\}}$ such that:
\begin{equation}\label{eqn:1}
\arraycolsep=2pt
\begin{array}{c}
     \model{\pder{\nu'}}(\vec{x}; \vec{y})
     = 
      g_{\nu'}(( \lambda \vec{u}\subseteq \vec{x}, \black{\lambda \vec v \subseteq \vec y}. \model{\pder{\mu}}(\vec{u};\vec{v}))_{\mu \in \close{\nu'}\cup \open{\nu'}})(\vec{x}; \vec{y})
\end{array}
\end{equation}
and, moreover, for all $\nu \in \close{\nu'}$:
\begin{equation}\label{eqn:2}
\arraycolsep=2pt
\begin{array}{c}
     \model{\pder \nu}(\vec{x};\vec y) = 
   g_{\nu}(( \lambda \vec{u}\subseteq \vec{x}, \black{\lambda \vec{v}\subseteq \vec y}. \model{\pder \mu}(\vec{u};\vec {v}))_{\mu \in \close{\nu'} \cup \open{\nu'}})(\vec{x}; \vec y)
   \end{array}
\end{equation}
Since $ \open{\nu'}=\open \hnu \cup X$ and  the path from $\nu'$ to any  $\mu\in X$ must cross an instance of $\cnd_{\sq}$ by  Proposition~\ref{prop:structureofcycles}.\ref{enum:budcompanion0}, the induction hypothesis on $\nu'$ (point~\ref{enum:b}) allows us to rewrite~\eqref{eqn:1}  as follows:
\begin{equation}\label{eqn:translation1}
\def\arraystretch{1.2}
\begin{array}{rl}
     \model{\pder{\nu'}}(\vec{x}; \vec{y})
    = 
     g_{\nu'}(&( \lambda \vec{u}\subseteq \vec{x},\black{\lambda \vec{v}\subseteq \vec y}. \model{\pder \nu}(\vec{u};\vec v))_{\nu \in \close{\nu'}}, \\ &( \lambda \vec{u}\subseteq \vec{x}, \black{\lambda \vec{v}\subseteq \vec y}. \model{\pder \mu}(\vec{u};\vec v))_{\mu \in  \open \hnu}, \\ &( \lambda \vec{u}\subset \vec{x}, \black{\lambda \vec{v}\subseteq \vec y}. \model{\pder{\mu}}(\vec{u};\vec v))_{\mu \in  X}\quad )(\vec{x}; \vec{y}) 
\end{array}
\end{equation}
On the other hand, for all $\nu \in \mathcal{C}_{\nu'}$,  for all  $\vec u \subseteq \vec x$ and $\vec{v}$, the equation in~\eqref{eqn:2} can be rewritten as:
\begin{equation}\label{eqn:translationintermediate}
\def\arraystretch{1.2}
   \begin{array}{rl}
\model{\pder \nu}(\vec{u};\vec v)
     = 
      g_{\nu}( &(\lambda \vec{w}\subset \vec{u}, \black{\lambda \vec {v'}\subseteq \vec v}. \model{\pder{\mu}}(\vec{u};\vec {v'}))_{\mu \in \close{\nu'}}, \\
  &   ( \lambda \vec{w}\subseteq \vec{u}, \black{\lambda \vec {v'}\subseteq \vec v}. \model{\pder{\mu}}(\vec{w};\vec {v'}))_{\mu \in  \open{\hnu}}, \\
  & ( \lambda \vec{w}\subseteq \vec{u}, \black{\lambda \vec {v'}\subseteq \vec v}. \model{\pder{\mu}}(\vec{w};\vec {v'}))_{\mu \in  X} \quad ) (\vec{u}; \vec v)
         \end{array}
\end{equation}
and so, for all $\nu \in \close{\nu'}$:
\begin{equation}\label{eqn:3}
\def\arraystretch{1.2}
\hspace{-0.3cm}
    \begin{array}{rl}
 \lambda \vec u \subseteq \vec x, \black{\lambda v\subseteq \vec y}. \model{\pder \nu}(\vec{u};\vec v)
     = 
     \lambda \vec u \subseteq \vec x, \black{\lambda v\subseteq \vec y}.  g_{\nu}( &(\lambda \vec{w}\subset \vec{x}, \black{\lambda \vec {v'}\subseteq \vec v}. \model{\pder{\mu}}(\vec{u};\vec {v'}))_{\mu \in \close{\nu'}}, \\
  &  ( \lambda \vec{w}\subseteq \vec{x}, \black{\lambda \vec {v'}\subseteq \vec v}. \model{\pder{\mu}}(\vec{w};\vec {v'}))_{\mu \in  \open{\hnu}}, \\
  & ( \lambda \vec{w}\subseteq \vec{x}, \black{\lambda \vec {v'}\subseteq \vec v}. \model{\pder{\mu}}(\vec{w};\vec {v'}))_{\mu \in  X} \quad ) (\vec{u}; \vec v)
         \end{array}
\end{equation}
Now, since  the paths from $\nu'$ to any  $\mu \in X$ in $\der$ must contain an instance of $\cnd_{\sq}$ \black{or an instance of $\cndl_\sq$}, for all  $\nu \in \close{\nu'}$ and all  $\mu \in X$, we have that either the path from $\nu'$ to $\nu$  contains an instance of $\cnd_{\sq}$ \black{or an instance of $\cndl_\sq$} or the path from $\nu$ to $\mu$ does. By applying the induction hypothesis on $\nu'$  (point~\ref{enum:b}), given  $\nu \in \close{\nu'}$  and   $\mu \in X$, either  $\lambda \vec{u}\subseteq \vec{x}, \black{\lambda \vec v\subseteq \vec y}. \model{\pder{\nu}}(\vec{u};\vec v)$ in~\eqref{eqn:translation1} is such that $\vec{u}\subset \vec{x}$,   or $\lambda \vec{w}\subseteq \vec{u}, \black{\lambda \vec {v'}\subseteq \vec v}. \model{\pder{\mu}}(\vec{w};\vec {v'})$ in~\eqref{eqn:translationintermediate} is such that $\vec{w}\subset \vec{u}$. This means that, for any $\mu \in X$,   $\lambda \vec{w}\subseteq \vec{x}, \black{\lambda \vec {v'}\subseteq \vec v}. \model{\pder{\mu}}(\vec{w};\vec {v'})$ in~\eqref{eqn:3} is such that $\vec{w}\subset \vec{x}$. For each  $\nu \in \close{\nu'}$, by rewriting $ \lambda \vec{u}\subseteq  \vec{x}, \black{\lambda \vec v \subseteq \vec y}. \model{\pder{\nu}}(\vec{u};\vec v)$  in~\eqref{eqn:translation1} according to the equation in~\eqref{eqn:3} we obtain:
\[
\begin{array}{c}
   \model{\pder{\hnu}}(\vec{x}; \vec{y})
   =
   t_{\hnu}( ( \lambda \vec{u}\subset \vec{x}, \black{\lambda \vec v\subseteq \vec y}. \model{\pder{\mu}}(\vec{u};\vec v))_{\mu \in  \close{\nu'}\cup X}, 
   ( \lambda \vec{u}\subseteq \vec{x}, \black{\lambda \vec v\subseteq \vec y}. \model{\pder{\mu}}(\vec{u};\vec v))_{\mu \in  \open{\hnu}})(\vec{x}; \vec{y})
\end{array}
\]
for some $t_{\hnu}$. Since  $\model{\pder{\mu}}= \model{\pder{\hnu}}$ for all $\mu \in X$,   and since  $\close{\hnu}= \close{\nu'}\cup \{ \hnu\}$, by setting $h_{\hnu}\dfn t_{\hnu}$ the above equation gives us the following:
\begin{equation}\label{eqn:translationintermediate2}
\begin{array}{c}
    \model{\pder{\hnu}}(\vec{x}; \vec{y})
    =
    h_{\hnu}(( \lambda \vec{u}\subseteq \vec{x}, \black{\lambda \vec v\subseteq \vec y}. \model{\pder{\mu}}(\vec{u};\vec v))_{\mu \in \close{\hnu}\cup \open{\hnu}})(\vec{x}; \vec{y})
    \end{array}
\end{equation}
which satisfies point~\ref{enum:b}.   From~\eqref{eqn:translationintermediate2} we are able to find the functions $(h_\nu)_{\nu \in \close \nu}$ defining the equations for $(\model{\pder \nu})_{\nu \in \close \nu}$. Indeed, the induction hypothesis on $\nu'$ gives us~\eqref{eqn:2} for any $\nu \in \close{\nu'}$. We rewrite in each such equation any   $\lambda \vec{u}\subseteq \vec{x}, \black{\lambda \vec{v}\subseteq y}.
\model{\pder{\mu}}(\vec{u};\vec v)$  such that  $\mu \in X$ according to equation~\eqref{eqn:translationintermediate2}, as $\model{\pder{\mu}}=\model{\pder{\hnu}}$ for any $\mu\in X$. We obtain the following equation for any $\nu \in \close{\nu'}$:
\[
\arraycolsep=2pt
\begin{array}{c}
\model{\pder{\nu}}(\vec{x}; \vec{y})
=
     t_\nu( ( \lambda \vec{u}\subset \vec{x}, \black{\lambda \vec v \subseteq \vec y}. \model{\pder{\mu}}(\vec{u};\vec v))_{\mu \in  \close{\nu'}\cup \{\hnu\}}, 
     ( \lambda \vec{u}\subseteq \vec{x}, \black{\lambda \vec v\subseteq \vec y}. \model{\pder{\mu}}(\vec{u};\vec v))_{\mu \in  \open{\hnu}})(\vec{x}; \vec{y})
\end{array}
\]
for some $t_\nu$. Since $\close{\hnu}=\close{\nu'}\cup \{\hnu\}$ and the above  equation satisfies point~\ref{enum:b}, we set $h_{\nu}\dfn t_{\nu}$ and we obtain, for all $\nu \in \close{\hnu}$:
\begin{equation}\label{eqn:4}
   \model{\pder{\nu}}(\vec{x}; \vec{y})= h_\nu( ( \lambda \vec{u}\subseteq \vec{x}, \black{\lambda \vec v\subseteq \vec y}. \model{\pder{\mu}}(\vec{u};\vec v))_{\mu \in  \close{\hnu}\cup \open{\hnu}})(\vec{x}; \vec{y})
\end{equation}
It remains to show that~\eqref{eqn:translationintermediate2} and~\eqref{eqn:4} satisfy point~\ref{enum:a}. On the one hand for all $\nu \in \close{\hnu}$ we  have  $h_{\nu}\in \bcpp((\model{\pder{\mu}})_{\mu \in \open{\hnu}}, (a_\mu)_{\mu \in \close{\hnu}}\black{,\safnor R} )$, with   $(a_\mu)_{\mu \in \close{\hnu}} $  oracle functions.  On the other hand, by applying the induction hypothesis, we have  $\model{\pder{\hnu}}=\model{\pder{\nu'}}\in \bcpp((\model{\pder\mu})_{\mu \in \open{\nu'}}\black{, \safnor R})$ and  $\model{\pder{\nu}}\in \bcpp((\model{\pder\mu})_{\mu \in \open{\nu'}}\black{, \safnor R})$, for all $\nu \in \close{\hnu}$.  Since $\open{\nu'}= \open{\hnu}\cup X$ and  $\model{\pder{\hnu}}= \model{\pder{\mu}}$ for all $\mu \in X$, we have both  $\model{\pder{\hnu}}\in \bcpp((\model{\pder\mu})_{\mu \in \open{\hnu}}\black{,  \safnor R})$ and   $\model{\pder{\nu}}\in \bcpp((\model{\pder\mu})_{\mu \in \open{\hnu}},  \model{\pder\hnu}\black{, \safnor R})$ for all $\nu \in \close{\hnu}$, and hence $\model{\pder{\nu}}\in \bcpp((\model{\pder\mu})_{\mu \in \open{\hnu}}\black{, \safnor R})$, for all $\nu \in \close{\hnu}$. 

Last, suppose that $\hnu$ is the conclusion of an instance of $\cut_{\n}$ with premises $\nu_1$ and $\nu_2$ then:
\begin{equation*}
    \model{\pder{\hnu}}(\vec{x}; \vec{y})= \model{\pder{\nu_2}}(\vec{x}; \model{\pder{\nu_1}}(\vec{x}; \vec{y}), \vec{y}) 
\end{equation*}
\black{By Proposition~\ref{prop:structureofcycles}.\ref{enum:budcompanion3} we have $\open{\nu_2}= \emptyset$, so that $\model{\pder{\nu_2}}(\vec{x};y, \vec{y})\in \bcpp(\safnor R)$ by applying the induction hypothesis on $\nu_1$}. We shall only consider the case where $\open{\nu_1}\neq \emptyset$. Then, \black{$\open{\hnu}=\open{\nu_1}$ and  $\close{\hnu}=\close{\nu_1}$}. 
Moreover,  induction hypothesis on $\nu_1$:
      \[
\arraycolsep=2pt
\begin{array}{c}
       \model{\pder{\nu_1}}(\vec{x};\vec{y})  = 
       h_{\nu_1}(( \lambda \vec{u}\subseteq \vec{x}, \black{\vec v\subseteq  \vec y}. \model{\pder{\mu}}(\vec{u};v,\vec v))_{\mu \in \close{\nu_1}\cup \open{\nu_1}})(\vec{x}; \vec{y})
\end{array}
\]
So that:
\begin{equation*}
     \model{\pder{\hnu}}(\vec{x}; \vec{y})=h_{\hnu}(( \lambda \vec{u}\subseteq \vec{x}, \black{\lambda \vec v\subseteq \vec y}. \model{\pder{\mu}}(\vec{u};\vec v))_{\mu \in \close{\hnu}\cup \open{\hnu}})(\vec{x}; \vec{y})
\end{equation*}
for some $h_{\hnu}$. Points~\ref{enum:a} and~\ref{enum:b}  hold by applying the induction hypothesis. 
\end{proof}

\section{Proof of~\Cref{lem:rel-bou-lem}}


Let us employ the notation $\sumlen{\vec x} \dfn \sum |\vec x|$ throughout this section.  We  will prove the following stronger statement. 
 
\begin{lem}[Bounding lemma, general version]
Let $f(\vec a)(\vec x;\vec y) \in \bcpp(\vec a, R)$, with $\vec a = a_1, \dots, a_k$ oracles and $R$ a set of relations. There is a function $m_f^{\vec c}(\vec x,\vec y)$ with,
$$
m_f^{\vec c}(\vec x,\vec y) = 
p_f(\sumlen{\vec x}) + \sum\vec c + \max |\vec y|
$$
for a monotone polynomial  $p_f(n)$ such that whenever there are constants $\vec c=c_1, \dots, c_k$ such that,
\begin{equation}
    \label{eq:oracle-const-max-bound}
    |a_i (\vec x_i;\vec y_i)| \leq c_i + \sum_{j\neq i} c_j + \max|\vec y_i|
\end{equation}
for $1\leq i \leq k$,
we have:\footnote{To be clear, here we write $|\vec y_i| \leq m_f^{\vec c}(\vec x,\vec y)  $ here as an abbreviation for $\{|y_{ij}| \leq m_f^{\vec c}(\vec x,\vec y)\}_j$.}
\begin{equation}
    \label{eq:elem-const-max-bound}
    |f(\vec a)(\vec x;\vec y)| \leq m_f^{\vec c}(\vec x,\vec y)
\end{equation}
\begin{equation}
    \label{eq:input-bounding-eqn}
    \begin{array}{rcl}
        f(\vec a)(\vec x; \vec y)    & = & f\left(\lambda \vec x_i.\lambda |\vec y_i| \leq m_f^{\vec c}(\vec x,\vec y) . a_i(\vec x_i;\vec y_i)\right)_i(\vec x;\vec y) \\
             f(\vec a)(\vec x; \vec y)   &  = & f\left(\lambda \vec x_r.\lambda |\vec y_r| \leq m_f^{\vec c}(\vec x,\vec y) . r(\vec x_r;\vec y_r)\right)_{r \in R}(\vec x;\vec y)
    \end{array}
\end{equation}
\end{lem}

\begin{rem}
Notice that~\Cref{lem:rel-bou-lem} is obtained by setting $\vec{a}=\emptyset$. Notice also that, in the case when $f(\vec x;\vec y)$ is just, say, $a_i(\vec x;\vec y)$, we may choose to set $\vec a = a_i$ or $\vec a = a_1,\dots, a_n$ in the above lemma, yielding different bounds in each case.
We shall exploit this in inductive hypotheses  in the proof that follows (typically when we write `WLoG').
\end{rem}

\begin{proof}
We prove~\Cref{eq:elem-const-max-bound} and~\Cref{eq:input-bounding-eqn} by induction on the definition of $f(\vec x;\vec y)$, always assuming that we have $\vec c$ satisfying \Cref{eq:oracle-const-max-bound}.

Let us start with~\Cref{eq:elem-const-max-bound}. If $f(\vec x;\vec y)$ is an initial function \black{or $f \in R$}  then it suffices to set $p_f(n) \dfn 1 +  n$.

If $f(\vec x;\vec y) = a_i(\vec x;\vec y)$ then it suffices to set $p_f(n) \dfn 0 $.

If $f(\vec x;\vec y) = h(\vec x;\vec y, g(\vec x;\vec y))$, let $p_h$ and $p_g$ be obtained from the inductive hypothesis. Notice that either $g$ or $h$ does not have oracles \black{in $\vec{a}$} by~\Cref{defn:bsubseteq}, so that:
$$
\arraycolsep=2pt
\begin{array}{rcl}
    |f(\vec x;\vec y)| & = & |h(\vec x; \vec y, g(\vec x; \vec y))| \\
    & \leq & p_h(\sumlen{\vec x}) + \sum\vec c + \max(|\vec y|, |g(\vec x;\vec y)|) \\
    & \leq & p_h(\sumlen{\vec x}) + \black{\delta_h}\sum\vec c + p_g(\sumlen{\vec x}) + \black{\delta_g}\sum \vec c + \max|\vec y|
\end{array}
$$
\black{where $\delta_g+\delta_h\in\{0,1\}$ such that  $\delta_g+\delta_h=1$}, so we may set $p_f(n) \dfn p_h(n) + p_g(n)$.

If $f(\vec x;\vec y) = h(\vec x, g(\vec x;);\vec y)$, let $p_h$ and $p_g$ be obtained from the inductive hypothesis. Note that, by definition of Safe Composition along a normal parameter, we must have that $g$ has no oracles \black{in $\vec a$ by~\Cref{defn:bsubseteq}}, and so in fact $|g(\vec x;)| \leq p_g(\sumlen{\vec x})$.
We thus have,
$$
\arraycolsep=2pt
\begin{array}{rcl}
    |f(\vec x;\vec y)| & = & |h(\vec x, g(\vec x;);\vec y)| \\
     & \leq & p_h(\sumlen{\vec x}, |g(\vec x;)|) + \sum \vec c + \max|\vec y| \\
     & \leq & p_h(\sumlen{\vec x}, p_g(\sumlen{\vec x})) + \sum \vec c + \max|\vec y| 
\end{array}
$$
so we may set $p_f(n)\dfn p_h(n, p_g(n))$.

Finally, if $f(\vec x;\vec y) = h(\lambda \vec u \permpref \vec x, \lambda \vec v\subseteq \vec y. f(\vec u;\vec v) ) (\vec x;\vec y)$, let $p_h$ be obtained from the inductive hypothesis.
We claim that it suffices to set  $p_f(n) \dfn n p_h(n)$. Notice that, by monotonicity:
\begin{equation}
 \label{eq:bnd-lem-rec-fn-inv}
    p_f(n) \geq p_h(n)+ p_f(n-1)
\end{equation}
for any $n>0$.


Now, to show  \Cref{eq:elem-const-max-bound} we proceed by a sub-induction on $\sumlen{\vec x}$.
For the base case, when $\sumlen{\vec x} = 0 $ (and so, indeed, $\vec x = \vec 0$), note simply that $\lambda \vec u\permpref \vec x, \lambda \vec v . f(\vec u;\vec v)$ is the constant function $0$, and so we may appeal to the main inductive hypothesis for $h(a)$ setting the corresponding constant $c$ for $a$ to be $0$ to obtain,
$$
\arraycolsep=2pt
\begin{array}{rcl}
    |f(\vec 0;\vec y)| & = & |h(0)(\vec 0;\vec y)| \\
        & \leq & p_h(0) + \sum\vec c + \max |\vec y| \\
        & \leq & p_f(0) + \sum\vec c + \max |\vec y|
\end{array}
$$
as required.
For the sub-inductive step,  let $\sumlen{\vec x} > 0$. 
Note that, whenever $\vec u \permpref \vec x$ we have $\sumlen{\vec u} < \sumlen{\vec x} $ and so, by the sub-inductive hypothesis and monotonicity of $p_f$ we have:
$$
|f(\vec u;\vec v)| \leq p_f(\sumlen {\vec x} -1) +  \sum \vec c + \max |\vec v|
$$
Now we may again appeal to the main inductive hypothesis for $h(a)$ by setting $c= p_f(\sumlen{\vec x}  -1)$ to be the corresponding constant for $a = \lambda \vec u \permpref \vec x, \lambda \vec v. f(\vec u;\vec v)$. 
We thus obtain:
$$
\arraycolsep=2pt
\begin{array}{rclll}
     |f(\vec x;\vec y)| & = & |h(\lambda \vec u \permpref \vec x, \lambda \vec v\subseteq \vec y . f(\vec u;\vec v)) (\vec x;\vec y)| \\
         & \leq & p_h (\sumlen{\vec x}) + c + \sum\vec c + \max |\vec y| & & \text{main IH} \\
         & \leq & ( p_h (\sumlen{\vec x}) + p_f(\sumlen{\vec x} -1) ) 
         + \sum\vec c + \max |\vec y| & & \text{def.~of $c$}\\
         & \leq & p_f(\sumlen{\vec x}) + \sum \vec c + \max|\vec y| && \eqref{eq:bnd-lem-rec-fn-inv}
\end{array}
$$
Let us now prove~\Cref{eq:input-bounding-eqn}, and let $p_f$ be constructed by induction on $f$ as above. We proceed again by induction on the definition of $f(\vec a)(\vec x;\vec y)$, always making explicit the oracles of a function.

The initial functions and oracle calls are immediate, due to the `$\max |\vec y|$' term in \Cref{eq:input-bounding-eqn}. 

If $f(\vec a)(\vec x; \vec y) = h(\vec a)(\vec x; \vec y, g(\vec a)(\vec x; \vec y))$ then, by the inductive hypothesis for $h(\vec a)$, any oracle call from $h(\vec a)$ only takes safe inputs of lengths:
    $$
    \begin{array}{rll}
        \leq & p_h(\sumlen{\vec x}) + \sum \vec c + \max (|\vec y|, |g(\vec a)(\vec x;\vec y)|) \\
        \leq & p_h(\sumlen{\vec x})+ \sum \vec c + p_g(\sumlen{\vec x}) + \sum \vec c + \max |\vec y| & \text{\eqref{eq:elem-const-max-bound}} \\
        \leq &  (p_h(\sumlen{\vec x} + p_g(\sumlen{\vec x})) + \sum\vec c + \max |\vec y|\\
         \leq & p_f(\sumlen{\vec x}) + \sum\vec c + \max|\vec y|
    \end{array}
    $$
    Note that any oracle call from $g(\vec a)$ will still only take safe inputs of lengths $\leq p_g(\sumlen{\vec x}) + \sum \vec c + \max |\vec y|$, by the inductive hypothesis, and $p_g$  is bounded above by $p_f$.
    
If $f(\vec a)(\vec x;\vec y) = h(\vec a)(\vec x, g(\emptyset)(\vec x;);\vec y)$ then, by the inductive hypothesis, any oracle call will only take safe inputs of lengths:
\[
\arraycolsep=2pt
\begin{array}{rlll}
    \leq & p_h(\sumlen{\vec x} + |g(\emptyset)(\vec x;)|) + \sum\vec c + \max |\vec y| \\
    \leq & p_h(\sumlen{\vec x} + p_g(\sumlen{\vec x})) + \sum\vec c + \max |\vec y| & &\eqref{eq:elem-const-max-bound}  \\
    \leq & p_f(\sumlen{\vec x}) + \sum \vec c + \max |\vec y|
\end{array}
\]

Last, suppose $f(\vec a)(\vec x;\vec y)= h(\vec a, \lambda \vec u\permpref \vec x, \lambda \vec v\subseteq \vec y . f(\vec a)(\vec u;\vec v))(\vec x;\vec y)$. 
We proceed by a sub-induction on $\sumlen{\vec x}$.
Note that, since $\vec u \permpref \vec x \implies \sumlen{\vec u} <\sumlen{\vec x}$, we immediately inherit from the inductive hypothesis the appropriate bound on safe inputs for oracle calls from $\lambda \vec u\permpref \vec x, \lambda \vec v\subseteq \vec y. f(\vec a)(\vec u;\vec v)$.

Now, recall from the proof of~\eqref{eq:elem-const-max-bound} that whenever $\vec u \permpref \vec x$ (and so $\sumlen{\vec u} < \sumlen{\vec x}$), we have $|f(\vec u;\vec v)| \leq p_f(\sumlen{\vec x}-1) + \sum \vec c + \max |\vec v|$.
So by setting $c = p_f(\sumlen{\vec x}-1)$ in the inductive hypothesis for $h(\vec a,a)$, with $a = \lambda \vec u\permpref \vec x, \lambda \vec v\subseteq \vec y. f(\vec a)(\vec u;\vec v)$, any oracle call from $h(\vec a,a)$ will only take safe inputs of lengths:
$$
\arraycolsep=2pt
\begin{array}{rlll}
    \leq & p_h(\sumlen{\vec x}) + p_f(\sumlen{\vec x}-1) + \sum \vec c + \max |\vec y| \\
    \leq & p_f(\sumlen{\vec x}) +  \sum \vec c + \max |\vec y| & & \text{\eqref{eq:bnd-lem-rec-fn-inv}}
\end{array}
$$
This completes the proof.  
\end{proof}

\section{Proof of~\Cref{prop:relativised-characterisation}}
In this section we prove~\Cref{prop:relativised-characterisation}, i.e.,  that $\bcpp(\safnor R) \subseteq \fptime(R)$, for $R$  a set of relations. We proceed by induction on the definition of $f(\vec x;\vec y)$.

Each initial function is polynomial-time computable, and each (relativised) complexity class considered is under composition, so it suffices to only consider the respective recursion schemes. Suppose we have $h(a)(\vec x;\vec y) \in \bcpp(a, \vec a\black{, \safnor R})$ and let:
$$
f(\vec x;\vec y) = h(\lambda \vec u \subset \vec x, \lambda \vec v \subseteq \vec y. f(\vec u; \vec v))(\vec x;\vec y)
$$
We start by making some observations:
\begin{enumerate}
    \item\label{item:poly-growth-rate} First, note that $|f(\vec x;\vec y)| \leq p_f(|\vec x|) +\sum \vec c + \max |\vec y| $, by~\Cref{lem:rel-bou-lem}, and so $|f(\vec x;\vec y)|$ is polynomial in $|\vec x,\vec y|$.
\item\label{item:poly-many-pp} Second, note that the set $[\vec x;\vec y] \dfn \{(\vec u,\vec v) \ | \ \vec u \permpref \vec x , \vec v \permprefeq \vec y\} $ has size polynomial in $|\vec x,\vec y|$:
\begin{itemize}
    \item write $\vec x = x_1 , \dots, x_m$ and $\vec y = y_1, \dots, y_n$.
    \item Each $x_i$ and $y_j$ have only linearly many prefixes, and so there are at most $|x_1|\cdot \cdots \cdot |x_m||y_1| \cdot \cdots \cdot |y_n| \leq \sumlen{\vec x,\vec y}^{m+n}$ many choices of prefixes for all the arguments $\vec x, \vec y$. 
    (This is a polynomial since $m$ and $n$ are global constants).
    \item Additionally, there are $m!$ permutations of the arguments $\vec x$ and $n!$ permutations of the arguments $\vec y$.
    Again, since $m$ and $n$ are global constants, we indeed have
    $|[\vec x;\vec y]| = O(\sumlen{\vec x, \vec y}^{m+n})$, which is polynomial in $|\vec x,\vec y|$.
\end{itemize}
\end{enumerate}

We describe a polynomial-time algorithm for computing $f(\vec x;\vec y)$ (over oracles $\vec a$\black{, $\safnor R$}) by a sort of `course-of-values' recursion on the order $\permpref \times \permprefeq $ on $[\vec x;\vec y]$.

First, for convenience, temporarily extend $\permpref \times \permprefeq$ to a total well-order on $[\vec x;\vec y]$, and write $S$ for the associated successor function.
    Note that $S$ can be computed in polynomial-time from $[\vec x;\vec y]$.

Define $F(\vec x,\vec y) \dfn \langle f(\vec u;\vec v)\rangle_{\vec u \permpref \vec x, \vec v \permprefeq \vec y}$, i.e.\ it is the graph of $\lambda \vec u \permpref \vec x, \lambda \vec v \permprefeq \vec y . f(\vec u;\vec v)$ that we shall use as a `lookup table'.
    Note that $|F(\vec x,\vec y)| $ is polynomial in $|\vec x,\vec y|$ by \Cref{item:poly-growth-rate} and \Cref{item:poly-many-pp} above.
    Now, we can write:\footnote{Here, as abuse of notation, we are now simply identifying $F(\vec x;\vec y)$ with $\lambda \vec u \permpref \vec x, \lambda \vec v \permprefeq \vec y . f(\vec x;\vec y)$.}
    \[
    \arraycolsep=2pt
    \begin{array}{c}
        F(S(\vec x,\vec y))  =  \langle f(S(\vec x,\vec y)), F(\vec x,\vec y) \rangle 
         =  \langle h(F(\vec x ,\vec y))(\vec x;\vec y),F(\vec x,\vec y)\rangle
    \end{array}
    \]
    Again by \Cref{item:poly-many-pp} (and since $F$ is polynomially bounded), this recursion terminates in polynomial-time.
    We may now simply calculate $f(\vec x;\vec y)$ as $h(F(\vec x,\vec y))(\vec x;\vec y)$.
    




\end{document}